\documentclass{vldb}

\usepackage{amsmath}
\usepackage{amssymb}
\usepackage[ruled, linesnumbered,lined,noend]{algorithm2e}
\usepackage{balance}
\usepackage{graphicx}
\usepackage{subfigure}
\usepackage[hidelinks]{hyperref}
\usepackage[capitalise]{cleveref}
\usepackage{txfonts}
\usepackage{color}
\usepackage{listings}
\numberofauthors{1}
\author{
\alignauthor
{Xuebin Su \qquad Hongzhi Wang \qquad Jianzhong Li \qquad Hong Gao}\\
\affaddr{Harbin Institute of Technology}\\
\email{xvebinsu@gmail.com \qquad \{wangzh, lijzh, honggao\}@hit.edu.cn}
}

\newtheorem{theorem}{Theorem}
\newtheorem{lemma}[theorem]{Lemma}

\newtheorem{proposition}[theorem]{Proposition}
\newdef{example}{Example}
\newdef{definition}{Definition}

\SetAlgorithmName{Algorithm}{Algorithm}{List of Algorithms}
\newcommand{\myref}[2]{\hyperref[#2]{#1~\ref{#2}}}

\newcommand{\prob}[1]{\operatorname{Pr}\left\lbrace#1\right\rbrace}
\newcommand{\expect}[1]{\operatorname{E} #1 }

\newcommand{\cov}[1]{\operatorname{Cov} #1 }

\newcommand{\dist}[1]{d\left(#1\right)}
\newcommand{\distext}[2]{d{#2}\left(#1\right)}
\newcommand{\irv}[1]{\operatorname{I}\left\lbrace#1\right\rbrace}
\newcommand{\func}[2]{\operatorname{\textsc{#1}}\left(#2\right)}
\newcommand{\trans}{^{\mathsf{T}}}
\newcommand{\argmax}{\operatorname*{arg\,max}}
\newcommand{\argmin}{\operatorname*{arg\,min}}
\newcommand{\tr}[1]{\operatorname{tr}\paren{#1}}
\newcommand{\mean}{\operatorname{\mathbb{E}}}

\newcommand{\paren}[1]{\left(#1\right)}
\newcommand{\bracket}[1]{\left[#1\right]}

\newcommand{\set}[1]{\left\lbrace#1\right\rbrace}
\newcommand{\norm}[1]{\left\lVert#1\right\rVert}
\newcommand{\seqsub}[3]{#1_1 #2 #1_2 #2 \ldots #2 #1_{#3}}
\newcommand{\cseqsub}[3]{#1_1 #2 #1_2 #2 \cdots #2 #1_{#3}}
\newcommand{\seqsup}[3]{#1^{(1)} #2 #1^{(2)} #2 \ldots #2 #1^{(#3)}}
\newcommand{\cseqsup}[3]{#1^{(1)} #2 #1^{(2)} #2 \cdots #2 #1^{(#3)}}
\renewcommand{\vec}[1]{\boldsymbol{\mathbf{#1}}}
\newcommand{\irow}[1]{\begin{pmatrix}#1\end{pmatrix}\trans}
\newcommand{\header}[1]{\noindent\textbf{#1}}

\newcommand{\mlterr}{\errsub{L^2}}
\newcommand{\sngerr}{\dist{\hat\theta, \theta}}
\newcommand{\errsub}[1]{\distext{\vec{\hat\theta}, \vec{\theta}}{_{#1}}}

\newcommand{\mltmod}{H \left(\vec{n}; \vec{\beta}\right)}
\newcommand{\modsup}[2]{H \left(\vec{n}^{(#1)}; \vec{\hat\beta}^{(#2)}\right)}

\begin{document}
\title{MISS: Finding Optimal Sample Sizes for Approximate Analytics}
\maketitle

\abstract{}
    Nowadays, sampling-based Approximate Query Processing (AQP) is widely
    regarded as a promising way to achieve interactivity in big data analytics.
    To build such an AQP system, finding the minimal sample size for a query
    regarding given error constraints in general, called Sample Size
    Optimization (SSO), is an essential yet unsolved problem. Ideally, the goal
    of solving the SSO problem is to achieve statistical accuracy, computational
    efficiency and broad applicability all at the same time. Existing approaches
    either make idealistic assumptions on the statistical properties of the
    query, or completely disregard them. This may result in overemphasizing only
    one of the three goals while neglect the others.

    To overcome these limitations, we first examine carefully the statistical
    properties shared by common analytical queries. Then, based on the
    properties, we propose a linear model describing the relationship between
    sample sizes and the approximation errors of a query, which is called the
    \emph{error model}. Then, we propose a Model-guided Iterative Sample
    Selection (MISS) framework to solve the SSO problem generally. Afterwards,
    based on the MISS framework, we propose a concrete algorithm, called
    \textsc{$L^2$Miss}, to find optimal sample sizes under the $L^2$ norm error
    metric. Moreover, we extend the \textsc{$L^2$Miss} algorithm to handle other
    error metrics. Finally, we show theoretically and empirically that the
    \textsc{$L^2$Miss} algorithm and its extensions achieve satisfactory
    accuracy and efficiency for a considerably wide range of analytical queries.
\endabstract{}

\section{Introduction}
\label{sec:introduction}

In the Big Data era, analyzing large volumes of data by analytical queries,
which summarize the data to discover useful information, becomes extremely
challenging for data scientists. This is mainly because under such circumstance,
it is infeasible to check the data record-by-record manually. Therefore, people
are always hoping to design systems to fully automate this process. However,
none of these efforts perfectly solves this problem. Therefore, recent years
have witnessed a surge of interest in designing novel systems to interact with
humans. By leveraging human knowledge, far more insights can be obtained from
the data with analytical queries.

However, the goal of incorporating human intelligence in analytical queries
brings new challenges. One of them that is the most essential yet elusive is
\emph{interactivity}~\cite{DBLP:conf/sigmod/Mozafari17}. To cope with it, one of
the most effective approach is to reduce the amount of data to be processed by
sampling, which is called sampling-based Approximate Query Processing
(AQP)~\cite{DBLP:conf/sigmod/Mozafari17}. The use cases for sampling-based AQP
typically include helping users to obtain a quick understanding of the data.
Such understanding may be inaccurate, but can be of great help for later
complicated and mission-critical analytical tasks.

When running an approximate analytical query, users would like to ensure that
the approximate result and the true one is almost the same. The difference
between these two results is called the approximation error, which is measured
by an \emph{error metric} chosen according to different scenarios, such as the
$L^2$ norm. To lower query latency and achieve interactivity for the query, we
hope to reduce the size of the sample to be processed by a given query as much
as possible. Therefore, one of the key problem in sampling-based AQP is Sample
Size Optimization (SSO),  i.e., to find the minimal sample size required to
answer an approximate query while ensuring that the approximation error
satisfies given constraints.

\header{Challenges:} When we consider solving the SSO problem, several
challenges arise due to its inherent uncertainty introduced by sampling. More
specifically, it is far from trivial to achieve all of the following goals.

\begin{list}{\labelitemi}{\leftmargin=1em}\itemsep 0pt \parskip 0pt
    \item \emph{Statistical accuracy,} which means that the resulting sample
    size is large enough to ensure that the approximate query result satisfies
    the user-defined error constraints. To ensure accuracy, error estimation
    methods are required to test whether the constraints hold. However, these
    methods often make strong assumptions on the query, which limit their range
    of application.

    \item \emph{Computational efficiency,} which means not only that the sample
    size should be as small as possible to maximize the speed of analytical
    queries, but also that SSO algorithms themselves should be efficient enough
    to avoid hampering the performance of the whole analytical process. To
    improve efficiency, the sample size should be as small as possible, which
    may increase the risk of being inaccurate.

    \item \emph{Broad applicability,} which means that the assumptions on both
    the data and the analytical functions should be weak enough to accommodate a
    wide variety of analytical tasks. To broaden the range of applications,
    error estimation methods based on weak assumptions are in demand. However,
    although adopting these methods might be beneficial in terms of
    applicability, it is often at the cost of efficiency since these methods are
    typically computationally intensive~\cite{DBLP:conf/sigmod/ZengGMZ14}.
\end{list}

\header{Limitations of existing methods:} From the discussions above, the three
goals of solving the SSO problem may contradict with each other. Thus, achieving
all of them at the same time seems impossible. Existing SSO methods often fall
short in one or more of the three aspects above. Specifically,

\begin{list}{\labelitemi}{\leftmargin=1em}\itemsep 0pt \parskip 0pt
    \item \emph{Statistical accuracy:} To estimate the approximation error, some
    methods~\cite{DBLP:conf/eurosys/AgarwalMPMMS13,
    DBLP:conf/sigmod/WangKFGKM14,DBLP:conf/sigmod/KrishnanWFGK16} assume that
    the sampling distribution is approximately normal such that the standard
    interval~\cite{diciccio1996} could be applied. However, such assumptions do
    not always hold, which brings the risk of producing inaccurate results.
    What's worse, this happens rather
    frequently~\cite{DBLP:conf/sigmod/AgarwalMKTJMMS14}.

    \item \emph{Computational efficiency:} Some
    methods~\cite{DBLP:journals/pvldb/KimBPIMR15, DBLP:conf/sigmod/DingHCC016,
    DBLP:conf/sigmod/Alabi016} rely on concentration
    inequalities~\cite{chung2006concentration} such as Hoeffding's
    inequality~\cite{Hoeffding1963} for error estimation. Such estimation is so
    conservative that it requires much larger sample size than necessary to meet
    the constraints~\cite{DBLP:conf/sigmod/AgarwalMKTJMMS14}, making it
    inefficient in terms of sample size. Some other
    methods~\cite{DBLP:conf/sigmod/ZengAS16} find optimal sample sizes in a
    mini-batch approach, which may result in a huge number of trials before the
    constraints are met and is inefficient in terms of the running time of the
    SSO method itself.

    \item \emph{Broad applicability:} Methods that exploit normality and
    concentration inequalities require assumptions on the analytical function to
    guarantee accuracy for a query. For example, the central limit theorem
    (CLT)~\cite{casella2002statistical} only works when the analytical function
    is \texttt{AVG}. This makes it hard to apply the SSO methods to arbitrarily
    complex analytical functions.
\end{list}

In summary, most existing methods overemphasize only one of the three goals
while paying little attention to the others by either making idealistic
assumptions on the statistical properties of the approximate analytical query or
almost completely ignoring them.

\header{Our contributions:} Being aware of the limitations, by using a novel
parametric model based on reasonably weak assumptions, we design a family of
different novel SSO methods under the same proposed framework but suitable for
various error metrics. The methods based on our framework not only have maximum
applicability being able to support almost all kinds of analytical queries, but
also achieve a decent balance between accuracy and efficiency by selecting
samples as small as possible while ensuring the error constraints to be
satisfied.
\begin{list}{\labelitemi}{\leftmargin=1em}\itemsep 0pt \parskip 0pt
    \item To maximize applicability, we propose a Model-guided Iterative Sample
    Selection (MISS) framework for generally addressing the SSO problem. It
    allow users to customize error metrics, sampling methods and error
    estimation methods for different scenarios and makes no assumption on the
    data and the query.

    \item To balance accuracy and efficiency, we propose a \emph{linear} error
    model describing the relationship between approximate errors and the sample
    sizes. Then, by combining the MISS framework and the error model, we propose
    an algorithm called \textsc{$L^2$Miss} that finds the optimal sample size
    for $L^2$ norm error. We show theoretically and empirically that
    \textsc{$L^2$Miss} is efficient while being accurate.

    \item To cope with different error metrics, we extend the \textsc{$L^2$Miss}
    algorithm under the MISS framework by converting the error constraints
    defined in terms of $L^2$ norm error to other metrics. Both theoretical and
    empirical study show that the extensions can handle a large variety of data
    and analytical functions efficiently, while providing satisfactory accuracy.
\end{list}

\header{Organization:} In \autoref{sec:pre}, we define the SSO problem formally
and develop the error model. In \autoref{sec:framework}, we propose MISS, a
general framework for solving the SSO problem. In \autoref{sec:l2miss} and
\autoref{sec:extensions}, we develop a series of algorithms under the MISS
framework to solve the SSO problem under various error metrics.  Experimental
evaluation is presented in \autoref{sec:experiments}, followed by related work
in \autoref{sec:related}. We state our conclusion and discuss future work in
\autoref{sec:conclusion}.

\section{Preliminaries}
\label{sec:pre}

\subsection{Problem Description}
\label{sec:problem}

In this paper, we mainly focus on analytical queries, which obtain summaries
from data. Following the convention in
BlinkDB~\cite{DBLP:conf/eurosys/AgarwalMPMMS13}, we define approximate
analytical queries as \autoref{lst:query}.
\begin{lstlisting}[caption={Approximate Analytical Queries}, label={lst:query}, language=SQL, mathescape=true, morekeywords={ERROR, WITHIN, CONFIDENCE}]
    SELECT $X, f(Y)$ FROM $D(X, Y)$
    GROUP BY $X$ WHERE $P$
    ERROR WITHIN $\epsilon$ CONFIDENCE $1 - \delta$
\end{lstlisting}

Here, $D$ denotes the dataset to perform the query on. $X$ and $Y$ are two sets
of attributes in $D$. Following the convention in
OLAP~\cite{DBLP:conf/sigmod/HellersteinHW97}, we call the attribute set $X$ the
group-by attributes, or \emph{dimensions}; and $Y$ the analytical attributes, or
\emph{measures}. $P$ denotes the selection predicate, which is intended for the
\texttt{COUNT}-with-predicate queries. $f$ is called the analytical function
computing the analytical results over $Y$ grouped by $X$.

We denote the true analytical result on $D$ and its approximation as
$\vec{\theta}$ and $\vec{\hat\theta}$ respectively. Note that for a $m$-group
query,  $\vec{\theta}$ and $\vec{\hat\theta}$ are both $m$-dimensional vectors,
each of which the entry corresponds to the query result of one group. To measure
the approximation error, we denote $d$ as the error metric, and the error can be
expressed as $\errsub{}$. To bound the error, we define error constraints to
have the form of
\begin{equation}
    \label{eq:errorConstraint}
    \prob{\errsub{} \leq \epsilon} \geq 1 - \delta,
\end{equation}
where $\epsilon$ is called the error bound, $\delta$ is called the error
probability, and $1 - \delta$ is called the confidence. From a statistical point
of view, \autoref{eq:errorConstraint} states that the \emph{margin of
error}~\cite{lohr2009sampling} of $\vec{\hat\theta}$ is no greater than
$\epsilon$ with confidence $1 - \delta$.

Note that throughout the paper, we assume that all the data are stored in a
single table. We refer readers to~\cite{DBLP:conf/sigmod/0001WYZ16} for the
approximate processing techniques of joins. We also would not consider handling
selections in a general circumstances other than \texttt{COUNT}-with-predicate
queries. For AQP of selections, we recommend readers to refer to
\cite{DBLP:conf/sigmod/DingHCC016} for more discussion.

In order to answer a query like in \autoref{lst:query} approximately, we draw a
sample $S \subseteq D$ from the dataset $D$, and apply the analytical function
$f$ on $S$ to obtain the approximate result. For an $m$-group query, we denote
the size of each group in $S$ and $D$ as $n_i$ and $|D|_i$, respectively, where
$i = 1,2,\ldots, m$. For convenience, we use $\vec{n}$ to denote the vector of
sample size, i.e., $\vec{n} = \irow{\cseqsub{n}{&}{m}}$. To find optimal
samples, we attempt to minimize the \emph{total sample size}, which is defined
as the sum of sample sizes of all groups, i.e.,
\begin{equation}
\label{eq:totalSize}
C\paren{\vec{n}} = \sum_{i = 1}^m n_i = \vec{1}\trans\vec{n}.
\end{equation}

Using the concept of error metric and total sample size, the SSO problem, which
aims to find the optimal sample size subject to the predefined error constraint,
is formalized as follows.

\begin{definition}
\label{def:ssoProblem}
    Given an approximate analytical query in \autoref{lst:query}, the sample
    size optimization problem is to
    \begin{equation}
    \label{eq:ssoProblem}
        \min_{\vec{n}} C(\vec{n}) \quad s.t. \quad \prob{\errsub{} \leq \epsilon} \geq 1 - \delta
    \end{equation}
    where $\vec{n}$ is the size of the sample $S$ drawn from $D$, $\epsilon$
    denotes the error bound and $\delta$ denotes the error probability.
\end{definition}

Solving the problem in \autoref{def:ssoProblem} is far from trivial. This is
because (i) the approximation error $\errsub{}$ is random, which makes the
problem stochastic rather than deterministic; and (ii) the closed-form
expression of $\errsub{}$ is unknown. The two reasons make it impossible to
apply the conventional optimization algorithms such as gradient-based ones to
find the optimum~\cite{DBLP:journals/corr/AmaranSSB17}. One effective way of
alleviating the two issues together is to approximate $\errsub{}$ with a
\emph{surrogate model}~\cite{box1987empirical}, which we call the \emph{error
model} to reduce the impact of noise and to guide the process of finding the
optimal sample size.

\subsection{Modeling Approximation Errors}
\label{sec:model}

Next, we would like to develop the error model, which takes a sample size
$\vec{n}$ as the input and outputs the predicted approximation error using
samples of size $\vec{n}$ to approximately answer the query.

\subsubsection{Statistical Properties of Analytical Queries}
To derive the model, we first examine the asymptotic properties shared by common
analytical queries, i.e., how approximation errors decrease as sample sizes
increase for typical queries. We then derive the error model based on such
properties.

In order to simplify the discussions, we assume that, queries only involve a
single group and denote the sample size, the true and the approximate analytical
result by $n$, $\theta$ and $\hat\theta$ respectively. Later, we will extend the
discussion to the multi-group case for completeness in this section.

\header{Relationship between error and sample size.} Intuitively, approximation
errors decrease as the corresponding sample sizes increase. To formalize the
intuition, We employ the concept of convergence in probability theory and
quantify how the error decreases using its \emph{rate of convergence in
probability}~\cite{van2000asymptotic}, which is defined as follows.
\begin{definition}
\label{def:errorConvergence}
Given the approximation error $\sngerr$ that converges in probability, the rate
of convergence in probability of $\sngerr$ is $O(r_n)$ if $r_n^{-1}\sngerr =
O_p(1)$, where $O_p(1)$ means ``bounded in probability''.
\end{definition}

Intuitively, the $O_p$ notation shares the same idea with the $O$ notation in
computational complexity in that for any two sequence $\left\lbrace a_n
\right\rbrace$ and  $\left\lbrace b_n \right\rbrace$, $a_n = O\left( b_n
\right)$ if and only if the ratio $a_n / b_n$ is bounded by a constant for a
sufficiently large $n$, while $a_n = O_p\left( b_n \right)$ means that the ratio
$a_n / b_n$ is bounded with arbitrarily large probability. Note that for
convenience, we omit the term ``in probability'' when talking about convergence
unless otherwise specified in the rest of this paper.

\header{Analytical results as statistics.} To find the the rate of convergence
of error for different types of queries using statistical theory, we claim that
for a sample $S = \set{\seqsub{X}{,}{n}}$, the results of a large variety of
analytical queries can be classified in to at least one of the following three
types of statistics:
\begin{list}{\labelitemi}{\leftmargin=1em}\itemsep 0pt \parskip 0pt
    \item \emph{U-statistics}~\cite{hoeffding1948}, which is of the form
    \begin{equation}
    \label{eq:uStatistics}
        U_n = \binom nk^{-1} \sum_\nu \kappa\left( X_{\nu_1}, X_{\nu_2}, \ldots, X_{\nu_k}\right),
    \end{equation}
    where $\set{\seqsub{\nu}{,}{k}}$ is a subset of $\set{1, 2, \ldots, n}$.
    Examples of U-statistics include common aggregate functions, e.g.,
    \texttt{AVG}, \texttt{VARIANCE}, and
    \texttt{PROPORTION}~\cite{van2000asymptotic}.

    \item \emph{M-estiamtors}~\cite{huber1964}, which is of the form
    \begin{equation}
    \label{eq:mEstimators}
        M_n = \argmax_\theta \sum_{i = 1}^n \psi_\theta \left(X_i\right),
    \end{equation}
    where $\theta$ is the parameter to be estimated, and $\psi_\theta$ are known
    functions. Examples of M-estimators include typical machine learning
    methods, such as linear regression, logistic regression and kernel
    regression~\cite{van2000asymptotic}.

    \item \emph{Inconsistent estimators}, of which the approximation errors do
    not converge to zero in probability, such that the error would not keep on
    decreasing with the increase of sample sizes. Examples of inconsistent
    estimators include some aggregate functions, such as \texttt{SUM} and
    \texttt{COUNT}~\cite{van2000asymptotic}.
\end{list}

\header{Rate of convergence of common statistics.} To derive the error model, we
examine the rate of convergence for common statistics, including the three
categories above and others.
\begin{list}{\labelitemi}{\leftmargin=1em}\itemsep 0pt \parskip 0pt
    \item The rate of convergence of U-statistics and M-estimators is given by
    the following lemma.
    \begin{lemma}
    \label{lem:convergence}
        Under weak conditions~\cite{van2000asymptotic}, the rate of convergence
        of U-statsitics $U_n$ and M-estimators $M_n$ are both $O\paren{n^{-b}}$,
        where $b$ is a positive constant.
    \end{lemma}

    \item For inconsistent estimators, it may be impossible to find the optimal
    sample size directly since the approximation errors may not continue to
    decrease no matter how the sample sizes increase. Fortunately, for some
    queries in this category, there exists a transformation converting them to
    consistent estimators, such as U-statistics and M-estimators. For example,
    \begin{equation*}
        \texttt{SUM} \left( Y \right) = |D| \cdot \texttt{AVG} \left( Y \right), \quad \texttt{COUNT} \left( Y \right) = |D| \cdot \texttt{PROPORTION} \left( Y \right)
    \end{equation*}
    where $|D|$ denotes the size of $D$. In such case, we can express the error
    constraints with respect to the corresponding consistent estimators using
    the transformation, and then optimize the sample size subject to the
    transformed constraints. And the actual error can be obtained by applying
    the inverse transformation.

    \item For some other statistics such as \texttt{QUANTILE} that are not
    U-statistics or M-estimators, they may be transformed into the two
    categories. Under specific conditions, these statistics may also converge at
    rate $O\paren{n^{-b}}$. Therefore, we claim that the convergence rate of
    many types of statistics, besides U-statistics and M-estimators, also has
    the form of $O\left( n^{-b} \right)$.
\end{list}

\header{Observation.} To summarize, we observe that, even though the concrete
expressions of analytical queries may differ, many of them still share one
principal asymptotic property, which is stated as the following proposition.
\begin{proposition}
\label{prop:asym}
    For common analytical queries, the approximation error $\sngerr$ converges
    at rate $O\left( n^{-b} \right)$.
\end{proposition}

\subsubsection{The Error Model}
\label{sec:mem}

Based on the observation above, we derive the error model to approximate the
relationship between sample sizes and errors.

\header{Single-group error model.} We first consider the single-group case.
According to \myref{Proposition}{prop:asym}, $\sngerr = O_P\left( n^{-b}
\right)$. Observe that in real-world applications, the absolute size of samples,
i.e., $n$ is typically large. Therefore, it is reasonable to ignore the
lower-order terms in the $O_P$ notation and only use the leading term.
Therefore, the single-group error is
\begin{equation}
\label{eq:singleModel}
    \sngerr \approx a n^{-b}
\end{equation}
where $a$, $b$ are constants. The logarithm of right hand side is called the
single-group error model, which is linear to $n$.

\header{Multi-group error model.} To derive the multi-group error model, we
first consider a specific error metric $d_g$, which is defined as the geometric
mean of the errors of all groups. Specifically, by plugging
\autoref{eq:singleModel} into the equation above, we have
\begin{equation*}
    \errsub{g} = \sqrt[m]{\prod_{i = 1}^m { d\paren{\hat\theta_i, \theta_i}}} = = \sqrt[m]{\prod_{i = 1}^m { a_i n_i^{-b_i}}}
\end{equation*}
where $d\paren{\hat\theta_i, \theta_i}$ is the error of group $i$, and $a_i$,
$b_i$ are all constants. Logarithms are then token on both side gives
\begin{equation*}
    \log \errsub{g} = \frac{1}{m} \paren{\sum_{i = 1}^m \log a_i - {\sum_{i = 1}^m { b_i n_i}}} = \beta_0 - \sum_{i = 1}^m \beta_i \log n_i,
\end{equation*}
where $\beta_i$ are all constants.

The right hand side of the second equation above, which is denoted by
$H\paren{\vec{n}; \vec{\beta}}$, is called the multi-group error model, or
simply the \emph{error model}, where $\vec\beta = \irow{\beta_0 & \beta_1 &
\cdots & \beta_m}$ is the parameter vector. For convenience, we denote
$\tilde{\vec n} = \irow{1 & -\log n_1 & -\log n_2 & \cdots & -\log n_m}$. The
model is then rewritten as
\begin{equation*}
    H\paren{\vec n; \vec\beta} = \vec \beta\trans \tilde{\vec n},
\end{equation*}
which is linear. Moreover, we use the notation $H\paren{\vec{n};
\vec{\hat\beta}^{(k)}}$ to denote the model with parameter $\vec\beta =
\vec{\hat\beta}^{(k)}$ and $H\paren{\vec{n}^{(k)}; \vec{\hat\beta}^{(k)}}$ to
denote the value of the model with parameter $\vec\beta = \vec{\hat\beta}^{(k)}$
at sample size $\vec n =\vec{n}^{(k)}$. We use $\vec{\hat\beta}$ to denote the
estimated value of $\vec{\beta}$.

Even though the error model $H\paren{\vec n; \vec\beta}$ is defined assuming
that the error metric is the geometric mean, as we will show in
\autoref{sec:l2miss} and \autoref{sec:extensions}, it performs well for other
commonly-used error metrics.

\section{Framework}
\label{sec:framework}

In this section, we propose the MISS framework to solve the SSO problem with as
few assumptions as possible. Specifically, we make no assumption on the data,
the query, the error metric, the sampling and the error estimation methods. By
doing so, we provide users a general enough approach to solve the problem.

Since the SSO problem involves black-box functions as constraints, finding an
exact closed-form solution is hard in general. Therefore, existing AQP systems
try to find an approximate solution using one of the following two approaches,
as discussed in \autoref{sec:introduction}.

\begin{list}{\labelitemi}{\leftmargin=1em}\itemsep 0pt \parskip 0pt
    \item \emph{The formula-based approach} ideally assumes that a closed-form
    approximation equation in terms of the sample size and the error is known.
    The predicted sample size can be directly obtained by solving the equation
    with respect to the desired error
    constraint~\cite{DBLP:conf/sigmod/AgarwalMKTJMMS14}.

    \item \emph{The model-free approach} exploits no priori knowledge of the
    relationship between sample size and approximation error, and simply guess
    the optimum by increasing the sample size by a small step in each iteration
    until the error constraint is satisfied~\cite{DBLP:conf/sigmod/ZengAS16}.
    The process of generating the guess in each iteration is call \emph{sample
    size searching}.
\end{list}

The formula-based approach is usually efficient enough by simply solving a
closed-form equation while the model-free approach is able to provide sufficient
accuracy guarantees and broadly applicable for a variety of queries by making
only a few assumptions. However, both approaches fail to balance accuracy,
efficiency and applicability. The former might suffer from inaccuracy or limited
applicability due to its idealistic assumptions, while the latter might suffer
from inefficiency due to the ignorance of the statistical property.

To make the best of their advantages while overcome their disadvantages, we
propose the Model-guided Iterative Sample Selection (MISS) framework to solve
the SSO problem generally. MISS iteratively estimates the error with respect to
the predicted sample size, polishes the model according to the of sample sizes
and errors observed, and makes the prediction again. 

This approach has two advantages. On one hand, by using a flexible model rather
than a fixed equation and iterative prediction, the risk of obtaining inaccurate
sample size can be greatly reduced compared to the formula-based approach. On
the other hand, this framework takes the advantage of the error model introduced
in \autoref{sec:model} to predict the optimal sample size such that the
unnecessary searching can be largely avoid compared to the model-free approach.

The basic idea of our MISS framework is to start with an initial guess of the
sample size, then to iteratively search for the optimum until the error
constraint is satisfied.

\begin{algorithm}
    \KwIn{Dataset $D$, analytical function $f$, error bound $\epsilon$ with
    error probability $\delta$, and error metric $d$.} \KwOut{Sample $S
    \subsetneqq D$ such that the approximation error $\mlterr \leq \epsilon$
    with probability $1 - \delta$.} \DontPrintSemicolon{} $P \leftarrow
    \emptyset$\; $\vec{n} \leftarrow
    \func{Initialize}{}$\;\label{ln:initialization} \While{True}{$S \leftarrow
    \func{Sample}{D, \vec{n}}$\;\label{ln:sampling} $e \leftarrow
    \func{Estimate}{S, d, f, \delta}$\;\label{ln:estimation} \lIf{$e \leq
    \epsilon$\label{ln:test}}{\Return{$S$}} \uElse{$P \leftarrow P \cup
    \set{\left( \vec{n}, e \right)}$\;\label{ln:errorProfile} $\vec{n}
    \leftarrow \func{Predict}{P, \epsilon}$\label{ln:prediction}}}
    \caption{The MISS Framework}
\label{alg:framework}
\end{algorithm}

The details of the MISS framework in shown in \autoref{alg:framework}. The MISS
framework first generates an initial guess of the sample size in
\myref{Line}{ln:initialization}. Then in each iteration, the framework draws a
sample of the generated size in \myref{Line}{ln:sampling}. Next, it estimates
the approximation error in terms of the sample in \myref{Line}{ln:estimation}.
Afterwards, it tests whether the error constraint is satisfied in
\myref{Line}{ln:test}. If the constraint is satisfied, it returns the selected
sample of optimal size successfully. Otherwise, in
\myref{Line}{ln:errorProfile}, the estimated error and the sample size is
collected as the \emph{error profile} defined as
\begin{equation}
    P = \set{\paren{\vec{n}^{(j)}, e^{(j)}} \middle| 1 \leq j\leq k}
\end{equation}
for prediction, where $l$ is the number of iterations. Within error profile,
$\vec{n}^{(j)}$ denotes the size of the sample in iteration $j$, and $e^{(j)}$
denotes the approximation error estimated in iteration $j$. Finally, it
generates the predicted optimal sample size using the error model in
\myref{Line}{ln:prediction}. With testing mechanism, this framework ensures to
terminate with the constraints satisfied.

The MISS framework is general enough to allow users to design the subroutines to
meet their needs, i.e., the \textsc{Initialize}, \textsc{Sample},
\textsc{Estimate}, and \textsc{Predict} subroutines. The interfaces and
functionalities of these subroutines are defined as follows.
\begin{list}{\labelitemi}{\leftmargin=1em}\itemsep 0pt \parskip 0pt
    \item \textsc{Initialize}: it generates an initial guess of the sample size.
    For a better prediction, the initial guess typically should better be a
    sequence of sample sizes rather than a single value. The goal of determining
    initial guesses is to make predictions more accurate and reliable without
    compromising much of efficiency. 

    \item \textsc{Sample}: it takes a sample $S$ randomly from the original
    dataset $D$ of size $\vec{n}$.

    \item \textsc{Estimate}: it estimates the error $e$ of the approximate
    result with error probability $\delta$ in terms of the analytical function
    $f$ evaluated on the sample $S$ measured by the error metric $d$. Typically,
    relaxing the assumptions required by error estimation methods helps to
    enlarge the range of applications of the SSO algorithm.

    \item \textsc{Predict}: it predicts the optimal sample size $\vec{n}$ using
    the error model. It uses the error profile $P$ to fit the model, then apply
    the model to find the optimum in terms of the total sample size $C$ subject
    to the error constraint. Specific models only work for specific categories
    of data, analytical functions and error metrics. When the model fails, the
    \textsc{Predict} subroutine recognizes the failure and returns an error.
\end{list}

According to the statistical properties and the performance requirements of the
query to be processed, data analysts are able to develop specific SSO algorithms
fit for their needs based on the MISS framework. As we will show in
\autoref{sec:l2miss}, for a specific scenario, by implementing the subroutines,
a concrete SSO algorithm can be derived specifically. Moreover, as shown in
\autoref{sec:extensions}, by extending the algorithm, we proposed a family of
algorithms that work under various error metrics.

\section{Finding Optimal Sample Sizes}
\label{sec:l2miss}

In this section, we propose a concrete SSO algorithm for the $L^2$ norm error
metric $d_{L^2}$, called \textsc{$L^2$Miss}, based on the MISS framework
proposed in \autoref{sec:framework} and utilizing the error model described in
\autoref{sec:model}. The $L^2$ norm error $\mlterr$ is defined as the $L^2$ norm
of approximation errors of all groups. Specifically,

\begin{equation}
    \mlterr = \sqrt{\sum_{i = 1}^m {\left( {\hat\theta}_i - {\theta}_i \right)}^2}
\end{equation}

Under the MISS framework, the \textsc{$L^2$Miss} algorithm is developed by
implementing the subroutines declared in \autoref{sec:framework}. To illustrate
the \textsc{$L^2$Miss} algorithm, we first introduce the sampling and the error
estimation methods in order. Then, we show how to predict the optimal sample
size in a model-guided approach. Finally, we present how to initialize the error
profile to make better predictions without compromising much of efficiency.

\subsection{Sampling}
\label{sec:sampling}

For sampling, To reduce the total sample size especially when the data size
differs substantially for each group~\cite{DBLP:journals/debu/MozafariN15}, we
use uniform stratified sampling, which takes a sample of size $n_i$ from each
group in $D$ uniformly at random. By controlling the sample size of each group
individually, the total size of samples can be largely reduced.

To obtain a stratified sample, almost all existing data management platforms,
such as Apache Spark~\cite{DBLP:conf/nsdi/ZahariaCDDMMFSS12}, require a full
scan. This is for two reasons. On one hand, the random sampling method offered
by them is typically Bernoulli sampling~\cite{DBLP:conf/sigmod/GryzGLZ04}, which
assigns a probability to each record to determine whether it should be included
in the sample. On the other hand, the \texttt{GROUP BY} clause requires to
examine the group-by attributes for each record to determine which group it
belongs to. However, full scans can be rather expensive when the data size is
large.

Fortunately, we are able to avoid full scans in sampling by adopting two
techniques: first, we use gap sampling~\cite{erlandson2014faster}, which assigns
a probability only to the records in the sample instead of to all records; then,
we use inverted index~\cite{DBLP:journals/pvldb/WangLHCPS17} on group-by
attributes to avoid full scans to examine group membership. Specifically, we
perform sampling on each inverted list for each combination of values of the
group-by attributes and obtain the stratified sample using the index.

\subsection{Error Estimation}
\label{sec:estimation}

For maximum applicability, we choose the bootstrap~\cite{diciccio1996} for
estimation. It does not make any assumption on the specific distribution of the
data $D$ or the specific analytical function $f$, while only requires some weak
regularity conditions to work as expected.

The bootstrap operates as follows~\cite{wasserman2006all}: to estimate the error
of a statistic $T$ computed by the analytical function $f$, it draws $B$
\emph{resamples} $S^*_b$, $b = 1, 2, \ldots, m$, from the original sample $S$
with replacement, each of the same size as $S$. For each resample, compute the
estimate $T^*_b$ on $S^*_b$ in the same way as computing $T$ on $S$. Then it
approximates the true sampling distribution of $T$, i.e., $F\paren{t} =
\prob{T\leq t}$ by empirical distribution $\mathbb{F}\paren{t} = \paren{1/
B}\sum_{b = 1}^B \irv{T^*_b \leq t}$, where $\operatorname{I}$ denotes the
indicator random variable. The $1 - \delta$ confidence region can be obtained by
taking the $1 - \delta$ \emph{quantile} of $\mathbb{F}\paren{t}$.

The correctness of the bootstrap depends on the consistency of approximating the
true distribution $F\paren{t}$ with the empirical distribution
$\mathbb{F}\paren{t}$. The conditions required by the bootstrap is given by the
following lemma.
\begin{lemma}
\label{lem:bootstrap}
    Let $\theta$ be the parameter to be estimated. Under weak
    assumptions~\cite{wasserman2006all}, the bootstrap estimates the
    approximation error with respect to $\theta$ consistently.
\end{lemma}
We call the difference between the estimated and the true approximation error
the \emph{estimation error}. The estimation error is introduced by the bootstrap
method and converges to $0$ in probability when the bootstrap is consistent. In
such cases, we are ensured that the bootstrap is performed correctly.

Even though the assumptions required by the bootstrap is much weaker compared to
other error estimation methods such as the normality-based ones, it is still
possible that the bootstrap fails. When this happens, we would have to seeking
alternative solution to provide approximate query results. Two typical cases
that the bootstrap is known to fail are as follows.
\begin{list}{\labelitemi}{\leftmargin=1em}\itemsep 0pt \parskip 0pt
\item \emph{Estimating \texttt{MIN} and
\texttt{MAX}}~\cite{DBLP:conf/sigmod/AgarwalMKTJMMS14}: In such cases, we
suggest to approximate \texttt{MIN} and \texttt{MAX} with the $\alpha$ and $1 -
\alpha$ quantiles respectively, of which the error can be correctly estimated by
the bootstrap, where $\alpha$ is a relatively small fraction.

\item \emph{Estimating heavy-tailed data}~\cite{wasserman2006all}: In such
cases, e.g., estimating the \texttt{AVG} of the
Pareto-distributed~\cite{casella2002statistical} data, we might have to turn to
concentration inequalities to estimate the error. 
\end{list}

As will be discussed in \autoref{sec:diagnostic}, we are able to develop
diagnostic methods to discover the cases that the bootstrap fails to be
consistent by checking the parameter of the model, and therefore can avoid
returning unreliable query results to users.

\subsection{Prediction}

\subsubsection{Applying the Model}
\label{sec:apply}

We now consider applying the error model $\mltmod$ in \autoref{sec:model} when
the error metric is $d_{L^2}$. Since the model is derived under error metric
$d_g$, we need to approximate $\errsub{L^2}$ with $\errsub{g}$. The following
theorem shows that such approximation is reasonable.

\begin{theorem}
\label{thm:modelError}
    $\forall \epsilon > 0$, if $\errsub{L^2} \leq \epsilon$, $\left|
    \errsub{L^2} - \errsub{g} \right| \leq \epsilon$.
\end{theorem}
\begin{proof}
   Since $0 \leq \errsub{g} \leq \max_{1 \leq i \leq m} \left| {\hat\theta}_i -
   {\theta}_i \right| \leq \errsub{L^2}$, we have $\left| \errsub{L^2} -
   \errsub{g} \right| \leq \errsub{L^2} \leq \epsilon$.
\end{proof}
We call $\left| \errsub{L^2} - \errsub{g} \right|$ the \emph{model error}. The
theorem implies that, when the approximation error $\errsub{L^2}$ is
sufficiently small, so is the model error. In real-world applications, users
typically would like to ensure that the approximation error is small enough. In
such cases, it is reasonable to apply the model.

\subsubsection{Model Fitting}

With the model $\mltmod$ at hand, we would like to figure out the best value of
the parameter $\mltmod$. To estimate the model parameter $\vec{\beta}$, we
attempt to minimize the MSE between the error profile observed and the model.
Therefore, in the $k$-th iteration of the algorithm, the problem of fitting
$\mltmod$ is formalized as
\begin{equation}
    \vec{\hat\beta} = \argmin_{\vec{\beta}} \norm{\tilde{N}\vec{\beta} - E}_2
\label{eq:olsFitting}
\end{equation}
where $E = \irow{\cseqsup{e}{&}{k}}$ and $\tilde{N} =
\irow{\cseqsup{\tilde{\vec{n}}}{&}{k}}$. By the normal
equation~\cite{nocedal2006numerical}, the solution of \autoref{eq:olsFitting},
denoted by $\vec{\hat\beta}_o$, is given by
\begin{equation}
\label{eq:olsSolution}
    \vec{\hat\beta}_o = \paren{\tilde{N}\trans\tilde{N}}^{-1}\tilde{N}\trans E.
\end{equation}

However, as mentioned in \autoref{sec:apply}, we actually approximate
$\errsub{L^2}$ with $\errsub{g}$. Therefore, to obtain a better model, we need
to calibrate our model to minimize the model error. According to
\autoref{thm:modelError}, we assign larger weights to error records with larger
sample sizes in the error profile $P$ since the corresponding model error is
smaller such that the overall model error is also smaller.

Therefore, to obtain a better the model, we penalize each record in the error
profile by multiplying each residual $\xi^{(k)}$ by a weight $w_k = \sum_{i =
1}^m n_i^{(k)}$, which is proportional to the sample size in regression. The
resulting problem is called weighted least square (WLS)
regression~\cite{wasserman2004all}. For that, we introduce a diagonal matrix
$W_{k\times k}$ with diagonal entries $\seqsub{w}{,}{k}$ such that the solution
of the weighted version of the fitting problem is given by
\begin{equation}
\label{eq:wlsSolution}
    \vec{\hat\beta}_w = \paren{\tilde{N}\trans W\tilde{N}}^{-1}\tilde{N}\trans WE.
\end{equation}
We use $\vec{\hat\beta}_w$ as the estimator of $\vec\beta$ in our algorithm and
denote it as $\vec{\hat\beta}$ for convenience in the remaining sections.

\subsubsection{Making Predictions}

After the model $\mltmod$ is fitted, we are able to apply the model to predict
the optimal sample size to avoid unnecessary searching overheads. Our goal is to
minimize the total sample size for a given sample size $\vec{n}$ subject to the
error constraint. In order to predict the optimal sample size, we solve the
approximate version of the SSO problem in \autoref{def:ssoProblem} by using the
model $\mltmod$ to approximate the error $\mlterr$, which is formalized as
\begin{equation}
\label{eq:prediction}
    \min_{\vec{n}} C\paren{\vec{n}} \quad \mathrm{s.t.} \quad \mltmod \leq \log \epsilon
\end{equation}
where $\epsilon$ is the user-defined error bound, the total sample size
$C\paren{\vec{n}} = \vec{1}\trans \vec{n}$ and the error model $\mltmod =
\vec{\beta}\trans \vec{\tilde{n}}$.

The approximate problem in \autoref{eq:prediction} is a constrained nonlinear
optimization problem, which has no closed-form solutions in general. However, by
using the method of Lagrange multipliers~\cite{nocedal2006numerical}, a
closed-form solution can be derived for this specific problem such that the
overhead of successive approximation required by conventional optimization
algorithms can be avoided.

To apply the method of Lagrange multipliers, we first define the Lagrange
function of the approximate problem as
\begin{equation*}
    L\paren{\vec{n}; \lambda}  = \vec{1}\trans \vec{n} - \lambda\paren{\vec{\beta}\trans \vec{\tilde{n}} - \log \epsilon}
\end{equation*}
such that the minimum of $L\paren{\vec{n}; \lambda}$ is exactly the solution of
the approximate problem in \autoref{eq:prediction}. Using calculus, finding the
minimum of $L\paren{\vec{n}; \lambda}$ is equivalent to solving the following
equations.
\begin{equation*}
    \frac{\partial L}{\partial n_i} = 1 - \lambda \beta_i\frac{1}{n_i} = 0,\quad
    \frac{\partial L}{\partial \lambda} = \vec{\beta}\trans \vec{\tilde{n}} - \log \epsilon = 0
\end{equation*}
where $i = 1, 2, \ldots, m$. Solving the equations above gives the predicted
optimal sample size $\vec{\hat n} = \irow{\cseqsub{\hat n}{&}{m}}$, where
\begin{equation}
\label{eq:predictSize}
    \hat{n_i} = \beta_i \exp \left( \frac{\beta_0 - \sum_{i = 1}^{m} \beta_i \log\beta_i - \log \epsilon}{\sum_{i = 1}^{m}\beta_i} \right).
\end{equation}
In real world, since we do not know $\vec{\beta}$ exactly, we use its estimator
$\vec{\hat \beta}$ in its place. Also, since sizes are all integers, we take the
nearest integer of $\hat{n_i}$ as our next guess of the sample size in the next
iteration of the algorithm.

\subsubsection{Failure Diagnostic}
\label{sec:diagnostic}

Even though the \textsc{$L^2$Miss} algorithm has a considerably wide range of
applications, it is still possible that the algorithm fails to estimate the
optimal sample size properly, which is mainly due to the following reasons:
\begin{enumerate}
    \item The predicted sample size is too large such that the algorithm may run
    out of resources (e.g., time and space).

    \item The predicted sample size does not keep increasing such that the error
    constraint can never be satisfied.
\end{enumerate}

Fortunately, using the error model, we are able to detect and recover from the
cases of failure beforehand. For that, we introduce a threshold $\tau > 0$ and
claim failures occur when $\sum_{i = 1}^m \hat\beta_i \leq \tau$ since we do not
know $\beta_i$ exactly. Such failures indicate that each $\beta_i$ is close to
$0$. Therefore, no matter how the sample size is increased, the approximation
error would almost not decrease. There are several reasons for such failures,
such as the inconsistency of the estimator or the estimated approximation error.
We call such failures to be \emph{unrecoverable}. In such cases, the algorithm
should raise an exception and exit.

Another type of failures is \emph{recoverable}. Recoverable failures are
typically caused by skewness, which means that for some groups, the increase of
the sample size has little or negative effect on error reduction. Recoverable
failures are indicated by the existence of some $\hat \beta_i < 0$. Such
failures would make the prediction subroutine fail to work properly since
\autoref{eq:predictSize} requires that each $\hat \beta_i > 0$ to take
logarithms.

Fortunately, we are able to recover from such failures by making all $\hat
\beta_i$ equal to eliminate negative values. Note that such adjustment would not
compromise accuracy. This is because even though the predicted sample size of
groups with negative or almost zero $\hat \beta_i$ is increased by increasing
$\hat \beta_i$ in \autoref{eq:predictSize}, it still has little effect in
reducing the error. As a result, the sample size of other groups would almost
not change in order to satisfy the error constraint. Therefore, the overall
sample size, i.e., the total sample size, would only increase, and the error
would only decrease as a consequence.

\begin{algorithm}
    \SetKw{Failure}{failure} \KwIn{Estimated parameter $\vec{\hat\beta}$ and
    threshold $\tau$.} \KwOut{The calibrated estimate $\vec{\hat\beta}$.}
    \DontPrintSemicolon{} \lIf{$\sum_{i = 1}^m \hat\beta_i \leq
    \tau$}{\Return{\Failure{}}}\uElseIf{$\min_{1 \leq i \leq m} \hat\beta_i \leq
    0$}{\For{$i \leftarrow 1$ \KwTo{} $m$}{$\hat\beta_i \leftarrow \sum_{i =
    1}^m \hat\beta_i / m$}} \Return{$\vec{\hat \beta}$}
    \caption{\textsc{Diagnostic}}
\label{alg:diagnostic}
\end{algorithm}

The diagnostic algorithm is formalized as \autoref{alg:diagnostic}. It first
determine whether an unrecoverable failure happens. If it does, the algorithm
exits abnormally with a failure to notify the user. Otherwise, it detects
whether the failure is recoverable. If so, the algorithm try to recover from the
failure by calibrating the estimated parameter. Otherwise, there is no failure,
then it returns the estimated parameter as it is.

\subsection{Initialization}

The initialization process is to generate a sequence of initial sample sizes $N$
whose length is $l$ such that the error profile can be initialized for model
fitting. To obtain a better model, we want the MSE of the estimator
$\vec{\hat\beta}$, i.e., $\expect \norm{\vec{\hat\beta} - \vec{\beta}}^2$ is
minimized. Since $\vec{\hat\beta}$ is unbiased, the MSE of $\vec{\hat\beta}$ is
proportional to the trace of the variance-covariance matrix of
$\vec{\hat\beta}$, which is denoted by
$\tr{\cov{\vec{\hat\beta}}}$~\cite{2007arXiv0707.0805C}. For simplification, we
assume here that all groups are mutually independent and all weights in $W$ are
equal such that for each group $i$ and sufficiently large $N_i$, minimizing
$\tr{\cov{\vec{\hat\beta}}}$ can be approximated by minimizing all $\pi_i =
\paren{\mean N_i}^2 / \mathbb{D}\, N_i$ where $N_i =
\irow{\cseqsup{n_i}{&}{l}}$, $\mathbb{E}\, N_i = (1 / l)\sum_{j = 1}^l n_i^{j}$,
and  $\mathbb{D}\, N_i = (1 / l)\sum_{j = 1}^l \paren{n_i^{j} - \mathbb{E}\,
N_i}^2$~\cite{wasserman2004all}.

To minimize $\pi_i$, we employ the Bhatia-Davis
inequality~\cite{bhatia2000better} that
\begin{equation*}
    \mathbb{D}\, N_i \leq \paren{\max N_i - \mean N_i} \paren{\mean N_i - \min N_i}.
\end{equation*}
Equality holds when all the $n_i^{(j)}$ are equal to either $\min N_i$ or $\max
N_i$ where $\max N_i$ and $\min N_i$ denote the maximum and the minimum of
$n_i^{(j)}$ in the vector $N_i$ respectively. Therefore, plugging the maximum of
$\mathbb{D}\, N_i$ into $\pi_i$ gives
\begin{equation*}
    \pi_i \geq \frac{\paren{\mean N_i}^2}{\paren{\max N_i - \mean N_i} \paren{\mean N_i - \min N_i}}.
\end{equation*}
Using calculus, the right hand side of the inequality above is minimized when
\begin{equation}
\label{eq:meanMinimum}
    \mean N_i = \frac{2}{\frac{1}{\min N_i} + \frac{1}{\max N_i}}
\end{equation}
To determine $\seqsup{n_i}{,}{l}$, suppose that $L_{\min}$ of all the
$n_i^{(j)}$ are equal to $\min N_i$ and the remaining $l_{\max} = l - l_{\min}$
of them are equal to $\max N_i$, then by the definition of $\mean N_i$,
\begin{equation}
\label{eq:meanDefinition}
    \mean N_i = \frac{l_{\min} \min N_i + l_{\max}\max N_i}{l}
\end{equation}
Combining \autoref{eq:meanMinimum} and \autoref{eq:meanDefinition} gives
\begin{equation}
\label{eq:sizeProportion}
    \frac{l_{\max}}{l_{\min}} = \frac{\min N_i}{\max N_i}
\end{equation}
From better performance, we cannot let $n_i^{(j)}$ become too large.
Therefore, we limit all the $n_i^{(j)}$ to be within the same interval $I_n
= \bracket{n_{\min}, n_{\max}}$. We call $I_n$ the \emph{initialization
interval}. By our argument above, $n_i^{(j)}$ should be equal to either
$n_{\min}$ or $n_{\max}$. According to \autoref{eq:sizeProportion}, we
sample each size $n_i^{(j)}$ from a distribution with probability
distribution $\Phi_n$ such that
\begin{equation}
\label{eq:initProb}
    \Phi_n\paren{n_{\min}} = \frac{n_{\max}}{n_{\min} + n_{\max}},\quad \Phi_n\paren{n_{\max}} = \frac{n_{\min}}{n_{\min} + n_{\max}}
\end{equation}
To obtain sample size of all groups, we repeat the process for $m$ times where
$m$ is the number of groups.

To determine $l$ and $I_n$, we suggest heuristically that $l$ should at least be
larger than $m + 1$ for regression while should not be too large for better
efficiency, $n_{min}$ should be large enough to make the bootstrap work
properly~\cite{hall1997bootstrap}, and $n_{max}$ should be orders of magnitude
smaller than the optimal size for better performance.

\subsection{The Algorithm and Analysis}

\subsubsection{Algorithm Description}
By combining the implementations of all the components of the MISS framework
described above, we now give the formal description of the \textsc{$L^2$Miss}
algorithm, which finds the optimal sample size satisfying the error constraint.

\begin{algorithm}
    \KwIn{Dataset $D$, analytical function $f$, error bound $\epsilon$ with
    error probability $\delta$, the number of bootstrap samples $B$,
    initialization interval $I_n = \bracket{n_{\min}, n_{\max}}$, and the length
    of initial sequence $l$.} \KwOut{A sample $S \subsetneqq D$ with optimal
    size.} \DontPrintSemicolon{} $k \leftarrow 1$\; Compute $\Phi_n$ by
    \autoref{eq:initProb} from $I_n$\; $m \leftarrow$ the number of groups in
    $D$\; \While{True}{\uIf{$k \leq l$}{$\vec n \leftarrow \func{Random}{\Phi_n,
    m}$\tcp*{initialization}\label{ln:random}} \Else{Compute $\vec{\hat{\beta}}$
    by \autoref{eq:wlsSolution}\label{ln:fit}\tcp*{model fitting} $\vec{\beta}
    \leftarrow \func{Diagnostic}{\vec{\hat{\beta}}, \tau}$ by
    \autoref{alg:diagnostic}\label{ln:diagnostic}\; $\vec{n} \leftarrow
    \vec{\hat n}$ by \autoref{eq:predictSize}\label{ln:optimize}\tcp*{making
    predictions}} $S \leftarrow \func{StratifiedSample}{D,
    \vec{n}}$\tcp*{sampling}\label{ln:strsample} $e \leftarrow
    \func{Bootstrap}{S, f, \delta, B}$\tcp*{error estimation}
        \label{ln:bootstrap}
        $P \leftarrow P \cup \set{\paren{\vec{n}, e}}$\label{ln:myprof}\;
        \lIf{$e \leq \epsilon$}{\Return{S}\label{ln:mytest}} $k \leftarrow k +
        1$\;}
    \caption{The \textsc{$L^2$Miss} algorithm}
\label{alg:l2Miss}
\end{algorithm}

The algorithm is shown in \autoref{alg:l2Miss}. The algorithm follows the
sample-estimate-predict loop as described in \autoref{alg:framework}. The
algorithm first generates a sample size $\vec{n}$. The mechanism of generating
sample sizes can be divided into two phases according to the value of $k$: (i)
in the first phase (\myref{Line}{ln:random}), when $k \leq l$, the algorithm
generates sample sizes randomly from the interval $I_n$ for initialization; and
(ii) in the second phase (\cref{ln:fit,ln:diagnostic,ln:optimize}), the
algorithm predicts the optimal sample size according to the error profile
constructed from previous observation. After obtaining the sample size
$\vec{n}$, it draws a stratified sample from the dataset $D$ of the generated
sample size $\vec{n}$ in \myref{Line}{ln:strsample}. Then, it estimates the
approximation error using bootstrapping in \myref{Line}{ln:bootstrap} and adds
the new error record $\paren{\vec{n}, e}$ to the error profile $P$ in
\myref{Line}{ln:myprof}. Next, it tests whether the error constraint is
satisfied in \myref{Line}{ln:mytest} as in \autoref{alg:framework} for the
current sample size using the prediction interval. If so, the algorithm returns
successfully with the required sample. Otherwise, the algorithm continues with
the next iteration.

For sample size prediction, \textsc{$L^2$Miss} adopts a model-guided approach,
i.e., it first fits the error model using \autoref{eq:wlsSolution}. Then it
diagnoses whether it fails and tries to recover from the failure. If a failure
happens and is unrecoverable, the algorithm exits abnormally. Otherwise, it
calibrate the parameter if necessary. Afterwards, it predict the optimal sample
size by minimizing the total sample size subject to the approximated error
constraint using the error model, i.e., to solve \autoref{eq:prediction}, and
use the result as the predicted sample size in the current iteration.

\subsubsection{Analysis}
\header{Accuracy:} We first show that the \textsc{$L^2$Miss} algorithm is
accurate, which means that it correctly finds the sample size satisfying the
predefined error constraint. To simplify our discussion, first note that if any
failure happens, it can always be detected by setting the threshold $\tau$
sufficiently large. Therefore, we assume that no failure happens in our proof
such that the estimated parameter $\seqsub{\hat\beta}{,}{m}$ are positive.

We first argue that, in each iteration $k$ in the prediction phase, i.e., $k
>l$, if the sample size $\vec{n}^{(k)}$ does not satisfy the error constraint,
which means that the model underestimates the error at $\vec{n}^{(k)}$, then in
the $(k + 1)$th iteration, the sample size $\vec{n}^{(k + 1)}$ would be
increased such that the error would continue to decrease. The above claim is
formally stated in the following lemma.
\begin{lemma}
\label{lem:size}
    In iteration $k > l$, if the sample size $\vec{n}^{(k)}$  does not satisfy
    the error constraint, then in iteration $(k + 1)$, the sample size
    $\vec{n}^{(k + 1)} > \vec{n}^{(k)}$, which means that for all $1 \leq i \leq
    m$, $n_i^{(k)} < n_i^{(k + 1)}$.
\end{lemma}
\begin{proof}
   We denote the estimated parameter in iteration $k$ and $(k + 1)$ as
   $\vec{\hat\beta}^{(k)}$ and $\vec{\hat\beta}^{(k + 1)}$ respectively. The
   difference between ${\hat\beta}^{(k)}$ and $\vec{\hat\beta}^{(k + 1)}$ is due
   to adding the error record $\paren{\vec{n^{(k)}}, e^{(k)}}$ to the error
   profile to fit the model. This is equivalent to first (i) add an error record
   $\paren{\vec{n}^{(k)}, \modsup{k}{k}}$ to the profile, then (ii) change the
   value the record to $\paren{\vec{n}^{(k)}, e^{(k)}}$.

   In step (i), note that it would not change $\vec{\hat\beta}^{(k)}$ since the
   residual at $\vec{n}^{(k)}$ is $0$, and $\modsup{k}{k} = \log \epsilon$ since
   $\vec{n}^{(k)}$ is optimal in iteration $k$. In step (ii), we perform linear
   regression with loss function $J\paren{\vec\beta} = \sum_{i = 1}^k
   {\xi^{(i)}}^2$. Since the error constraint is not satisfied, we have $\log
   e^{(k)} > \log \epsilon = \modsup{k}{k}$. We assume that the residual
   $\xi^{(k)} = \modsup{k}{k} - \log e^{(k)} < 0$ without loss of generality,
   the partial derivative $\partial J / \partial \xi^{(k)} = 2\xi^{(k)} < 0$. It
   means to minimize $J\paren{\vec{\beta}}$, $\xi^{(k)}$ should be increased,
   and so as $\modsup{k}{k + 1}$ when $\vec\beta = \vec{\hat\beta}^{(k + 1)}$
   compared to $\vec{\hat\beta}^{(k)}$. Afterwards, to satisfy the error
   constraint, the algorithm predicts $\vec{n}^{(k + 1)}$ where the error is
   predicted to be smaller than the one at $\vec{n}^{(k)}$. To show that
   $\vec{n}^{(k + 1)} > \vec{n}^{(k)}$, note that in the model $H\paren{\vec{n};
   \vec\beta^{(k + 1)}}$, $\hat\beta_i > 0$ and $\partial H / \partial n_i^{(k +
   1)} = \hat\beta_i^{(k + 1)} > 0$ for every group $i = 1, 2, \ldots, m$, which
   means that to let the values $H\paren{\vec n; \vec\beta^{(k + 1)}}$ decrease
   from where $\vec n = \vec{n}^{(k)}$ to where $\vec n = \vec{n}^{(k + 1)}$,
   there exists some $1 \leq i \leq m$ such that $n_i^{(k + 1)} > n_i^{(k)}$.
   Moreover, by \autoref{eq:predictSize}, for all $j = 1, 2, \ldots, m$, $n_i /
   n_j = \hat\beta_i^{(k + 1)} / \hat\beta_j^{(k + 1)} > 0$. Therefore, all
   $n_j$ also would be increased. In conclusion, $\vec{n}^{(k + 1)} >
   \vec{n}^{(k)}$.
\end{proof}

Using the lemma above, we are able to show the correctness of the
\textsc{$L^2$Miss} algorithm by arguing that, the algorithm will keep increasing
the sample size such that the error keeps decreasing until the error constraint
is unsatisfied. Therefore, the algorithm finds the optimal sample size with
regards to the error constraint, which is stated in \autoref{thm:correct}. We
omit the proof due to the limit of space.

\begin{theorem}
\label{thm:correct}
    Given that no failure happens, \autoref{alg:l2Miss} correctly finds the
    optimal sample size satisfying the given error constraint.
\end{theorem}
\begin{proof} In each iteration, the algorithm tests whether the error
   constraint is satisfied by comparing the error bound $\epsilon$ with the
   estimated approximation error $e$. If the constraint is satisfied, then it
   returns and the theorem is trivially correct. Otherwise, since we assume
   that there is no failure, which implies that the error $e$ is correctly
   estimated and converges to $0$ in probability. Then there exists $\vec{n}'
   > 0$ such that the corresponding error at $\vec{n}$ satisfies the error
   constraint, i.e., $e' < \epsilon$. By \myref{Lemma}{lem:size}, since the
   sample size continues to increase in every iteration in the prediction
   phase, there must exists a certain iteration $K > l$ such that
   $\vec{n}^{(K)} > \vec{n}'$, which implies that $e^{(K)} < e' < \epsilon$.
   Therefore, the error constraint is satisfied in iteration at most $K$.
   \end{proof}

\header{Efficiency:} We then show that the the \textsc{$L^2$Miss} algorithm is
efficient in terms of both total sample size and computational complexity. As
before, we assume that no failure happens to simplify the discussion.

In terms of total sample size, we claim that the \textsc{$L^2$Miss} algorithm
finds near-optimal sample sizes satisfying the error constraint. To show that,
we first argue that the the difference between any predicted sample size
$\vec{n}^{(k)}$ and the optimal size $\vec{n}'$ can be arbitrarily small when
$\vec{n}^{(k)}$ becomes sufficiently large. This is because that the difference
is caused by mainly the following two factor: (i) the model error introduced by
approximating $\errsub{L^2}$ with $\errsub{g}$, and (ii) the estimation error
introduced by the bootstrap method. By \autoref{thm:modelError} and the
consistency of the bootstrap, both types of error converge to $0$ in
probability.

In terms of computational complexity, we show that the \textsc{$L^2$Miss}
algorithm is efficient. First we claim that the number of iterations of the
\textsc{$L^2$Miss} algorithm is upper bounded by the difference of two predicted
sample sizes in the beginning and at the end of the prediction phase, which is
formally stated in the following lemma.
\begin{lemma}
\label{lem:upper}
    Let $K$ be the total number of iterations in \autoref{alg:l2Miss}. Then $K
    \leq l + 1 + \vec{n}^{(K)} - \vec{n}^{(l + 1)}$, where $l$ is the length of
    the initial sample size sequence, and the difference of two sample sizes in
    any pair of iterations $i, j$ is defined as $\vec{n}^{(i)} - \vec{n}^{(j)} =
    \max_{1 \leq i \leq m}\paren{n^{(i)}_i - n^{(j)}_i}$.
\end{lemma}
\begin{proof} Note that $\vec{n}^{(l + 1)}$ and $\vec{n}^{(K)}$ is the first
   and the last sample size generated in the prediction phase in
   \autoref{alg:l2Miss}. As a result, $K \geq l + 1$, which implies that
   $\vec{n}^{(K)} > \vec{n}^{(l + 1)}$ by \myref{Lemma}{lem:size}. To obtain
   an upper bound of $K - l$, which is exactly the number of predictions
   \autoref{alg:l2Miss} makes, note that by \myref{Lemma}{lem:size}, in each
   iteration of the prediction phase, the sample size of each group increases
   by at least $1$. Therefore, to reach $\vec{n}^{(K)}$ from $\vec{n}^{(l +
   1)}$ would require at most $\vec{n}^{(K)} - \vec{n}^{(l + 1)}$ iterations,
   i.e., $K - l \leq \vec{n}^{(K)} - \vec{n}^{(l + 1)} + 1$, which completes
   the proof. \end{proof}

The difference between any two predicted sample size, in particular,
$\vec{n}^{(K)}$ and $\vec{n}^{(l + 1)}$, is no larger than twice of the maximum
of all the difference between the predicted and the optimal sample size.
Therefore, as discussed above, $\vec{n}^{(K)} - \vec{n}^{(l + 1)}$ can also be
arbitrarily small when the sample size is sufficiently large. As a result, the
number of iterations in the prediction phase $K - l$ can be also arbitrarily
small, which means that, with only a few predictions, the algorithm is able to
find a near-optimal sample size satisfying the error constraint.

Using the upper bound of the number of iterations $K$ in
\myref{Lemma}{lem:upper}, we give the computational complexity of the
\textsc{$L^2$Miss} algorithm in the following theorem.
\begin{theorem}
    Let $K$ be the total number of iterations in \autoref{alg:l2Miss}. The
    expected running time of \autoref{alg:l2Miss} is
    \begin{equation}
    \label{eq:bound}
        O\paren{ B\paren{ l \cdot m \cdot n_{\max} + \paren{K - l} C\paren{\vec{n}^{(K)}} } + \sum_{k = l + 1}^K Q_{k, m} }
    \end{equation}
    where $O\paren{Q_{k, m}}$ is the worst-case running time for $k-$observation
    and $m-$variate linear regression~\cite{wasserman2006all}.
\end{theorem}
\begin{proof} In iteration $k$ of \autoref{alg:l2Miss}, the sampling method
   takes $O\paren{C\paren{\vec{n}^{(k)}}}$ expected time using gap
   sampling~\cite{erlandson2014faster} with indexing, the \textsc{Bootstrap}
   method takes $O\paren{B \cdot C\paren{\vec{n}^{(k)}}}$ time in the worst
   case where $B$ is the number of bootstrap samples. The algorithm takes
   $O\paren{Q_{k, m}}$ time to fit the parameter and $O\paren{m}$ time for
   diagnostic and generating predictions. In general, $O\paren{m} \subseteq
   O\paren{Q_{k, m}}$ Therefore, the algorithm takes $O\paren{B \cdot
   C\paren{\vec{n}^{(k)}} + Q_{k, m}}$ time in iteration $k$.

   To figure out the total time for \autoref{alg:l2Miss}, note that out of the
   $K$ iterations, the first $l$ iterations are in the initialization phase,
   where the total sample size $C\paren{\vec{n}^{(k)}} \leq m \cdot n_{\max}$,
   while the last $K - l$ iterations are in the prediction phase, where the
   total sample size $C\paren{\vec{n}^{(k)}} \leq C\paren{\vec{n}^{(K)}}$  by
   \myref{Lemma}{lem:size}. Therefore, substituting $C\paren{\vec{n}^{(k)}}$
   and summing all iterations gives exactly \autoref{eq:bound} in the theorem.
   \end{proof}

The theorem above shows that the algorithm is computationally efficient, since
its running time only depends on the sample size, which is near-optimal, rather
than the size of the entire data.

\section{Extensions}
\label{sec:extensions}

In this section, we extend the \textsc{$L^2$Miss} algorithm to accommodate other
error metrics under the MISS framework.

\subsection{Basic Idea}

Suppose that the new error constraint to be satisfied is defined in terms of
error metric $d'$ with error bound $\epsilon'$, while our \textsc{$L^2$Miss}
algorithm is defined in terms of $L^2$ norm error, denoted by $d_{L^2}$ with
error bound $\epsilon$. The following lemma gives a necessary condition that the
error constraint in terms of $d'$ holds.

\begin{lemma}
    \label{lem:extension}
    Suppose, for error probability $\delta$, $\prob{\mlterr \leq \epsilon'} \geq
    1 - \delta$, then $\prob{d'\paren{\vec{\hat\theta}, \vec{\theta}} \leq
    \epsilon} \geq 1 - \delta$ if $R \subseteq R'$.
    \begin{equation}
        R = \set{\vec{v} \middle| d'\paren{\vec{v}, \vec{\theta}} \leq \epsilon}, \quad R' = \set{\vec{v} \middle| d_{L^2}\paren{\vec{v}, \vec{\theta}} \leq \epsilon'}
    \end{equation}
\end{lemma}

\begin{proof} 
    By the definition of probability, $R \subseteq R'$ implies that
   $\prob{\vec{\hat\theta} \in R'} \geq \prob{\vec{\hat\theta} \in R} \geq 1 -
   \delta$. 
\end{proof}

\myref{Lemma}{lem:extension} implies that, in order to find the optimal sample
size with regards to $d'$, all we need is to find the equivalent error bound
$\epsilon'$ such that $R \subseteq R'$, and then call the \textsc{$L^2$Miss}
algorithm to find the optimal sample size with error bound $\epsilon'$. The
process is formally stated in \autoref{alg:extension}.
\begin{algorithm}
    \KwIn{User-defined error bound $\epsilon'$ and the error bound conversion
    function $\Gamma$ defined interms of $d'$} \KwOut{A sample $S$ of optimal
    size.} \DontPrintSemicolon{} $\epsilon' \leftarrow \Gamma\paren{\epsilon}$\;
    Return the result of calling \textsc{$L^2$Miss} with error bound
    $\epsilon'$\;
    \caption{Extending the \textsc{$L^2$Miss} algorithm}
    \label{alg:extension}
\end{algorithm}

We define an error bound conversion function $\Gamma$, which converts the
user-given error bound $\epsilon$ in terms of the new metric $d'$ to the
equivalent error bound $\epsilon'$ in terms of $d_{L^2}$ such that $R \subseteq
R'$ in \myref{Lemma}{lem:extension} holds. Once we know $d'$ specifically,
finding $\Gamma$ is not hard.

\autoref{alg:extension} only serves as a framework. In order to obtain a
concrete sample size optimization algorithm, we need to find the error bound
conversion function $\Gamma$. In the remaining of this section, we will show in
detail how to extend \textsc{$L^2$Miss} to work with other widely-used error
metrics by finding suitable $\Gamma$.

\subsection{Maximum Error}

One error metric that is also commonly used is the maximum of errors of all
groups. With such metric, users are guaranteed that the approximation error is
no more than some predefined bounds. The maximum error of $\vec{\hat\theta}$
with respect to the true one $\vec{\theta}$ is defined as
\begin{equation*}
    \errsub{L^\infty} = \max_{1 \leq i \leq m}\paren{\hat\theta_i - \theta_i},
\end{equation*}
which is exactly the $L^\infty$ norm of the errors.

To convert a given error bound $\epsilon$ in terms of $d_{L^\infty}$ into the
error bound $\epsilon'$ such that $\errsub{L^\infty} \leq \epsilon$ if
$\errsub{L^2} \leq \epsilon'$, we define the conversion function as
$\Gamma\paren{\epsilon} = \epsilon$. We denote the algorithm that finds optimal
sample sizes for $d_{L^\infty}$  as \textsc{MaxMiss}. To show the correctness of
the \textsc{MaxMiss} algorithm, we have the following theorem.
\begin{theorem}
    For all $\epsilon > 0$, $\errsub{L^\infty} \leq \epsilon$ if $\errsub{L^2}
    \leq \epsilon$.
\label{thm:maxBound}
\end{theorem}

\begin{proof}
   Immediate from the fact that the $L^2$ norm is no larger than the $L^\infty$
   norm, i.e., $\max_{1 \leq i \leq m} \paren{{\hat\theta}_i - {\theta}_i} \leq
   \sqrt{\sum_{i = 1}^{m}{\paren{{\hat\theta}_i - {\theta}_i}^2}} $.
\end{proof}

Moreover, to support other $L^p$ norms as error metrics, denoted by $d_{L^p}$
where $p \geq 1$, observe that for all $p > 2$, $\errsub{L^p} \leq \errsub{L^2}$
and $\errsub{L^1} \leq \sqrt{m} \cdot \errsub{L^2}$ where $m$ is the number of
groups~\cite{kreyszig1989introductory}. By taking advantage of such
inequalities, the corresponding conversion functions can be derived trivially.

\subsection{Ordering}
\label{sec:ordering}

In many scenarios, we want approximate results to preserve the same order as the
true one in order to, for example, answer approximate Top-$k$ queries or to
visualize analytical results with accurate
trends\cite{DBLP:journals/pvldb/KimBPIMR15}. Such requirement can be formalize
as the \emph{correct-ordering property}~\cite{DBLP:journals/pvldb/KimBPIMR15},
which is defined as follows.
\begin{definition}
    Suppose that ${\theta}_{\eta_1} \leq {\theta}_{\eta_2} \leq \cdots \leq
    {\theta}_{\eta_m}$ for a true query result $\vec{\theta}$, the approximate
    query result $\vec{\hat\theta}$ preserves the correct-ordering property with
    respect to $\vec{\theta}$ if ${\hat\theta}_{\eta_1} \leq
    {\hat\theta}_{\eta_2} \leq \cdots \leq {\hat\theta}_{\eta_m}$ where
    $\seqsub{\eta}{,}{m}$ is a permutation of $1, 2, \ldots, m$.
\label{def:ordering}
\end{definition}

In order to provide ordering guarantee using the MISS framework, we consider
ordering as an implicit error metric and the correct-ordering property as an
error bound defined in terms of $\vec{\hat\theta}$, i.e.,
$\epsilon\paren{\vec{\hat\theta}}$, instead of a predefined constant. The reason
of doing so is that the difference between $\vec{\theta}$ and $\vec{\hat\theta}$
needs to be small enough to preserve such property.

To find the conversion function $\Gamma$, we treat $\vec{\theta}$ as a point in
a $m$-dimensional space. The coordinates are denoted by $\seqsub{x}{,}{m}$. We
define the conversion function $\Gamma\paren{\epsilon\paren{\vec{\hat\theta}}} =
\min_{1 \leq i < j \leq m}\rho_{ij}$ where $\rho_{ij}$ denotes the distance from
the point $\vec{\hat\theta}$ to the hyperplane $x_i = x_j$. We denote the
algorithm that finds optimal sample sizes preserving the correct-ordering
property as \textsc{OrderMiss}.

The following theorem claims the correctness of the \textsc{OrderMiss}
algorithm.
\begin{theorem}
\label{thm:ordering}
    $\vec{\hat\theta}$ preserves the correct-ordering property if
    \begin{equation}
        \errsub{L^2} \leq \min_{1 \leq i < j \leq m}\rho_{ij}
    \label{eq:orderCondition}
    \end{equation}
    where $\rho_{ij}$ denotes the distance from $\vec{\hat\theta}$ to the
    hyperplane $x_i = x_j$.
\end{theorem}

\begin{proof}
   Observe that the $\binom m 2$ hyperplanes $x_i = x_j,\ 1\leq i < j \leq m$,
   divide the m-dimensional space into $2\binom m 2$ subspaces, each of which
   determines a unique order of the $m$ coordinates of a point in the space.
   Therefore, $\vec{\hat\theta}$ preserves the same order as $\vec{\theta}$ if
   and only if they fall into the same subspaces.

   We prove that $\vec{\hat\theta}$ and $\vec{\theta}$ are in the same subsapce
   if $\errsub{L^2} \leq \min_{1 \leq i < j \leq m}\rho_{ij}$ by contradiction.
   Suppose that $\vec{\hat\theta}$ and $\vec{\theta}$ are in different
   subspaces. Then there exists a hyperplane $x_{\eta_i} = x_{\eta_j}, 1\leq
   \eta_i < \eta_j \leq m$, such that $\vec{\hat\theta}$ and $\vec{\theta}$ are
   on the different sides of $x_{\eta_i} = x_{\eta_j}$. Therefore $\errsub{L^2}
   > \rho_{\eta_i\eta_j}$, which contradicts with \autoref{eq:orderCondition}.
\end{proof}

In order to apply the error bound conversion function, one naive algorithm is to
enumerate all pairs of $i, j$, compute the distance $\rho_{ij}$ by its
definition and find the minimum, which takes $O\paren{m^2}$ time where $m$ is
the number of groups. We propose a more efficient algorithm called
\textsc{OrderBound} that finds the error bound in $O\paren{m\log m}$ time if a
comparison-based sorting algorithm is used. \textsc{OrderBound} first sorts the
$m$ entries of $\vec{\hat\theta}$, then returns the minimum of the difference of
every two adjacent entries of the sorted vector as the minimum of difference of
all pairs of entries. The details of the \textsc{OrderBound} algorithm is shown
in \autoref{alg:orderBound}.
\begin{algorithm}
    \KwIn{Approximate query result $\vec{\hat\theta}$.} \KwOut{Error bound
    $\epsilon'$ such that $\vec{\hat\theta}$ preserves the correct ordering
    property if $\errsub{L^2} \leq \epsilon'$.} \DontPrintSemicolon
    $\vec{\hat\theta}^* = \func{Sort}{\vec{\hat\theta}}$
    \tcp*{${\hat\theta}^*_{i + 1} \geq {\hat\theta}^*_i, \,i = 1, 2, \ldots, m$}
    $\rho \leftarrow$ new $(m-1)$-dimensional vector\; \For{$i \leftarrow 1$
    \KwTo{} $m - 1$}{$\rho_i \leftarrow {\hat\theta}_{i + 1}^* -
    {\hat\theta}_i^*$} \Return{$\left( \min_{1 \leq i \leq m -1}{\rho_i} \right)
    / \sqrt{2}$}
    \caption{\textsc{OrderBound}}
\label{alg:orderBound}
\end{algorithm}

Since $\vec{\hat\theta}$ is random, to make the \textsc{OrderMiss} algorithm
more reliable, it is beneficial to repeat computing $\vec{\hat\theta}$ for
multiple times on different samples and take the average before using
\autoref{alg:orderBound} for error bound conversion.

The correctness of the \textsc{OrderBound} algorithm follows the theorem below.
\begin{theorem}
\label{thm:orderBound}
    Suppose ${\hat\theta}_{\eta_1} \leq {\hat\theta}_{\eta_2} \leq \cdots \leq
    {\hat\theta}_{\eta_m}$ where $\seqsub{\eta}{,}{m}$ is a permutation of $1,
    2, \ldots, m$, then
    \begin{equation*}
        \min_{1 \leq i < j \leq m}\rho_{ij} = \frac{1}{\sqrt{2}} \min_{1 \leq k \leq m - 1}\left( {\hat\theta}_{\eta_{k + 1}} - {\hat\theta}_{\eta_k} \right)
    \end{equation*}
\end{theorem}

\begin{proof} By definition, the distance $\rho_{ij}$ is the projection of
   $\vec{\hat{\theta}}$ onto the normal vector of the hyperplane $x_i =
   x_j$~\cite{boyd2004convex}. Therefore, $\rho_{ij} = \left| {\hat\theta}_i -
   {\hat\theta}_i \right| / \sqrt{2}$, which implies $\min_{i, j}\rho_{ij} =
   \min_{i, j}\left| {\hat\theta}_i - {\hat\theta}_j \right| / \sqrt{2}$.
   Suppose for contradiction that $\eta_k, \eta_k' = \argmin_{i,j} \left|
   {\hat\theta}_i - {\hat\theta}_j \right|$ and ${\hat\theta}_{\eta_k} \leq
   {\hat\theta}_{\eta_k'}$ while $\eta_k' \neq \eta_{k + 1}$ without loss of
   generality. Then by definition, ${\hat\theta}_{\eta_k} \leq
   {\hat\theta}_{\eta_{k + 1}} \leq {\hat\theta}_{\eta_k'}$. Therefore, $
   {\hat\theta}_{\eta_{k + 1}} - {\hat\theta}_{\eta_k} \leq
   {\hat\theta}_{\eta_k'} - {\hat\theta}_{\eta_k}$, which contradicts with the
   definition of $\eta_k, \eta_k'$. \end{proof}

\subsection{Maximal Difference Error}

Sometimes, ensuring the correct-ordering property is not enough for analytical
tasks such as visualization. This is because the difference between every pair
of groups of approximate results might differ largely, such that it is difficult
to draw valid quantitative conclusions on the trends of some parameters of the
data even though they preserves the same order as the true results.

Therefore, a property stronger than the correct-ordering property is required to
quantify trends. We define the approximate result $\vec{\hat\theta}$ to have
bounded maximal difference error if and only if the maximum of the approximation
error of any two groups is bounded, which is formalized as follows.
\begin{definition}
    The maximum difference error of an approximate query result
    $\vec{\hat\theta}$ with respect to the true result $\vec{\theta}$ is defined
    as
    \begin{equation}
        \errsub{\Delta} = \max_{1 \leq i, j \leq m} \left| \paren{{\hat\theta}_i - {\hat\theta}_j} - \paren{{\theta}_i - {\theta}_j}  \right|
    \end{equation}
\end{definition}

To convert an error bound in terms of $d_{\Delta}$ to the error bound in terms
of $d_{L^2}$, we define the conversion function $\Gamma\paren{\epsilon} =
\epsilon / \sqrt{2}$. The algorithm that finds the optimal sample sizes for
$d_{\Delta}$ is denoted as \textsc{DiffMiss}. The correctness of the
\textsc{DiffMiss} algorithm follows the following theorem. 

\begin{theorem}
    $\forall \epsilon > 0$, $\errsub{\Delta} \leq \epsilon$ if $\errsub{L^2}
    \leq \epsilon / \sqrt{2}$.
\end{theorem}

\begin{proof}
   Since the objective function $\errsub{\Delta}$ and the feasible region
   $\errsub{L^2} \leq \epsilon / \sqrt{2}$ are symmetric, the maxima and the
   minima of all $\Delta_{ij} = \paren{{\hat\theta}_i - {\hat\theta}_j} -
   \paren{{\theta}_i - {\theta}_j}$ are equal where $1 \leq i < j \leq m$. To
   find the maximum and the minimum of $\Delta_{ij}$, we use the method of
   Lagrange multipliers and optimizing the Lagrange function
   $L\paren{\vec{\hat\theta};\lambda} = \Delta_{ij} - \lambda\paren{\errsub{L^2}
   - \epsilon / \sqrt{2}}$. Solving the equations that $\partial L / \partial\,
   {\hat\theta}_i = 0, \,  i = 1, 2, \ldots, m$ and $\partial L / \partial
   \lambda = 0$ gives that the maximum and the minimum of $\Delta_{ij}$ is
   $\epsilon, -\epsilon$ respectively. Therefore the maximum of $\left|
   \Delta_{ij} \right|$ is $\epsilon$, which implies that $\errsub{\Delta} \leq
   \epsilon$.
\end{proof}

\section{Experiments}
\label{sec:experiments}

In this section, we evaluate the \textsc{$L^2$Miss} algorithm as its extensions
empirically and compare them with the state-of-the-art AQP algorithms. All the
experiments are performed on a Linux server with one 16-core, 32-thread CPU,
32GB memory and 2TB hard disk space.

\subsection{Evaluation Criteria}

We first introduce the criteria in the experiments to evaluate our algorithms as
well as other AQP algorithms in terms of both accuracy and efficiency
respectively.

For accuracy, we introduce \emph{simulated confidence} as an indicator of
whether the error constraint is satisfied. The simulated confidence for an
approximate query result is computed as follows: for a specific sample size
given by an AQP system, we draw a large number, i.e. $1000$ in general, of
samples of such size, compute the analytical results on each sample. The
simulated confidence, denoted by $\hat c$, is defined as the frequency that the
results satisfy the error bound. We claim that an AQP algorithm is accurate if
$\hat c \geq 1 - \delta$, meaning that the simulated confidence is no smaller
than the user-defined confidence.

For efficiency, we introduce two indicators, i.e., the running time and the
total sample size $C\paren{\vec{n}}$. Longer running time or larger total sample
size imply that the algorithm is less efficient. Since our algorithms are
stochastic, we repeat the experimental process for several times and report the
average and the standard deviation for both the running time and the total
sample size. Larger standard deviation indicates that the resulting sample size
may lie far away from the optimal one, making the algorithm less reliable.

Furthermore, we also would like to evaluate the effectiveness of our error
model. For that, we use the \emph{coefficient of
determination}~\cite{wasserman2006all}, also called the $r^2$ score, to measure
the \emph{goodness of fit} of a model. The range of $r^2$ score is $\left(
-\infty, 1 \right]$. The larger $r^2$ is, the better the model fits the data
(i.e., the error profile in this case).

\subsection{Applicability Evaluation}

In this section, we first demonstrate the broad applicability of the proposed
\textsc{$L^2$Miss} algorithm. The applicability of other MISS-based algorithms,
including \textsc{MaxMiss}, \textsc{NormalMiss}, and \textsc{DiffMiss}, is the
same since they do nothing but only call \textsc{$L^2$Miss} with different
parameters. We test our algorithm against different data distributions and
analytical functions to see whether the algorithm returns sample sizes
accurately, i.e., satisfying the user-defined error constraint. Specifically, we
consider totally six common analytical functions in the experiments,  including
not only aggregate functions supported by traditional databases, including mean
(\texttt{AVG}), variance (\texttt{VAR}), median (\texttt{MEDIAN}), and maximum
(\texttt{MAX}), but also widely-used machine learning algorithms such as linear
regression (\texttt{LINREG}) and logistic regression (\texttt{LOGREG}). We
consider the following data distributions: the standard normal distribution
(Normal), the exponential distribution with scale $1$ (Exp), the uniform
distribution on range $\left[0, 1\right]$ (Uniform), and the Pareto distribution
with scale $1, 2$ and $3$ (Pareto1, Pareto2, Pareto3).

For some analytical functions, such as \texttt{MEDIAN} and \texttt{MAX}, and
some data distributions, such as Pareto1 and Pareto2, there is no closed-form
error estimation method that can be applied. And some others, for example,
\texttt{MAX}, Pareto1, and Pareto2, the bootstrap is not able to provide
consistent error estimation, as pointed out in \autoref{sec:estimation}.
Therefore, in the following experiments, we will test our algorithm against such
extreme cases in order to see whether the results are consistent with the
theory.

\subsubsection{Analytical Functions and Data Distributions}
\label{sec:single}
In the first part, we evaluate our \textsc{$L^2$Miss} algorithm for various
combinations of the analytical functions and data distributions described above.
We denote each cases by a function-distribution name pair. For each case,we
examine the simulated confidence $\hat c$ and the $r^2$ score to see the
accuracy of the algorithm and the effectiveness the model respectively.

Moreover, we set the parameter of \textsc{$L^2$Miss} in \autoref{alg:l2Miss} in
the following way. $\epsilon$ is set to be a relative error bound $\epsilon^*$
times the value of the true analytical result, which is computed prior to the
experiment, and $\epsilon^*$ is set to $0.05$ for \texttt{LOGREG} and $0.01$ for
other analytical functions. We set $\delta = 0.05$, $B = 500$, $I_n =
\bracket{4000, 8000}$, and $l = 20$. The number of tuples in this experiment is
$100$ million.

\begin{figure*}[ht]
    \centering
    \includegraphics{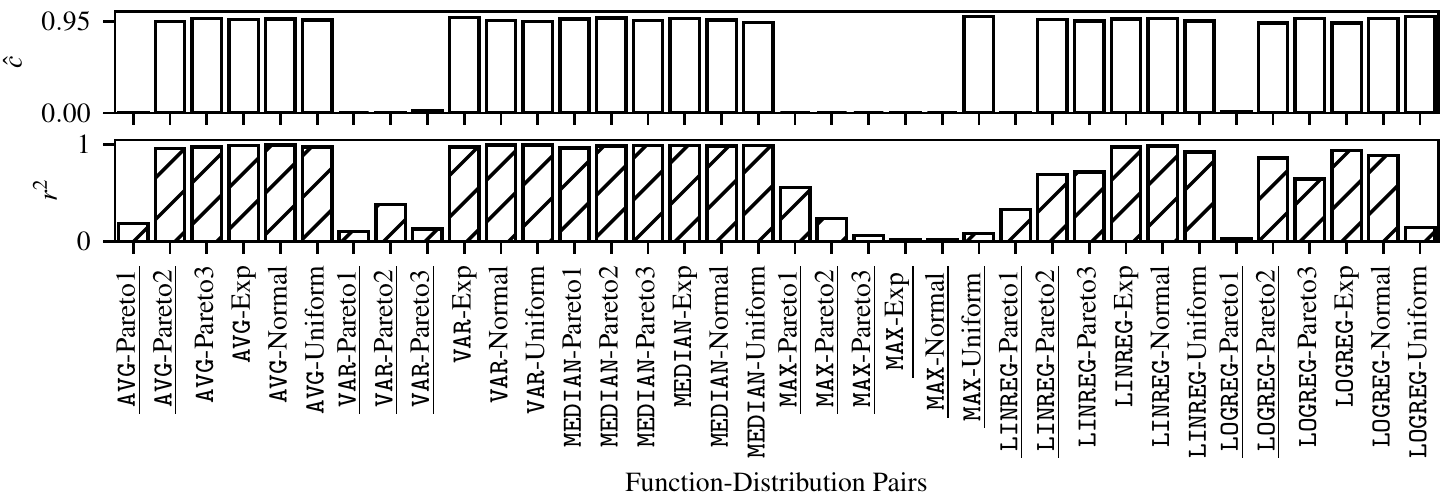}
    \caption{Applicability Evaluation for Function-Distribution Pairs. Cases where the bootstrap is theoretically inconsistent are underlined.}
\label{fig:single}
\end{figure*}

\autoref{fig:single} shows our results for all function-distribution pairs
described above. The function-distribution pairs underlined are the cases where
the bootstrap cannot estimate the approximation error correctly (i.e.,
\myref{Lemma}{lem:bootstrap} does not apply). The upper plot shows the simulated
confidence for all cases, and the lower plot demonstrates the corresponding
$r^2$ scores. Ideally, the simulated confidence should be close to $1 - \delta =
0.95$, and the $r^2$ score should be close to $1$ when the algorithm performs
satisfactorily. Otherwise, the simulated confidence would be far away from
$0.95$, which indicates that the algorithm produces sample size either too small
to be accurate, or too large to be efficient, and the $r^2$ score would be far
away from $1$.

As seen in the results, our \textsc{$L^2$Miss} algorithm demonstrates its board
applicability on a range of data generating distributions and analytical
functions. Specifically, the algorithm finds optimal sample size accurately for
$25$ out of $36$ cases. When the bootstrap guarantee to be consistent, i.e., the
$21$ cases that are not underlined, the simulated confidence of our algorithm is
close to $0.95$ and the $r^2$ score is close to $1$, meaning that it finds the
optimal sample sizes accurately for these cases. For the cases that the
bootstrap is not consistent, no closed-form method can be applied, while our
algorithm is still accurate for the $4$ out of $15$ cases, for example,
\texttt{AVG}-Pareto2. This is because not only that the variance of the
approximate analytical result is not infinity since the data size is finite, but
also that our algorithm uses an iterative approach to minimized the impact of
noise.

\subsubsection{Multi-group Data}
\label{sec:multi}
In the second part, we evaluate our algorithm on the data consisting of two
independent groups generated by two distributions respectively. Each case is
denoted by a distribution-distribution name pair. The distributions used here
are the same as that described in the beginning of this section, i.e., Pareto1,
Pareto2, Pareto3, Exp, Normal, and Uniform. We choose the analytical function to
be \texttt{AVG}, which is the most commonly used one. The data size is set to be
$100$ million tuples for each group. For the parameters of the
\textsc{$L^2$Miss} algorithm, we choose the relative error bound $\epsilon^* =
0.01$ and others remain the same as the previous experiment described in
\autoref{sec:single}.

\begin{figure}[ht]
    \centering
    \includegraphics{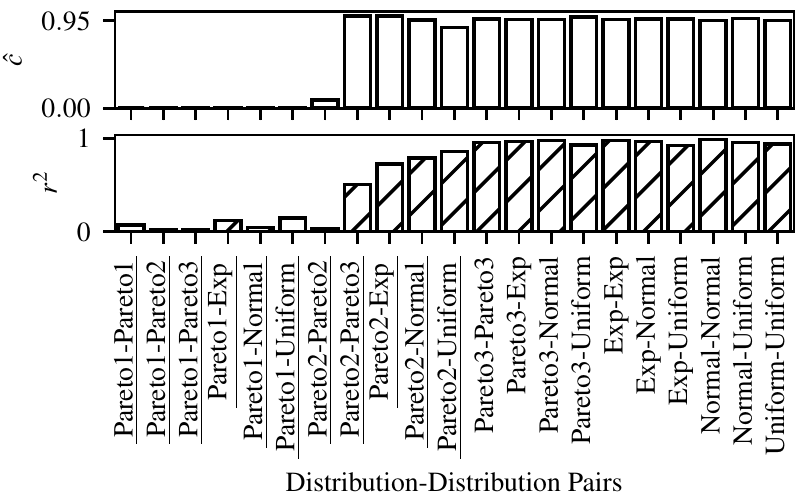}
    \caption{Applicability Evaluation for Averages of Distribution Pairs. Cases where the bootstrap is theoretically inconsistent are underlined.}
\label{fig:multi}
\end{figure}

\autoref{fig:multi} shows the results of the experiment for all pairs of
distributions. As described in \autoref{sec:single}, the underlined cases are
the ones that bootstrap is not consistent theoretically. Since the two groups
are generated independently, the bootstrap estimates the error consistently if
and only if it is consistent for each group of the pair. Also, like in
\autoref{fig:single}, we plot the simulated confidence $\hat c$ and the $r^2$
score for all cases to figure out whether our algorithm produces the accurate
sample sizes and whether our model depicts the relationship between the
approximation error and the sample size effectively.

As shown in \autoref{fig:multi}, the \textsc{$L^2$Miss} algorithm demonstrates
its capability to handle the $L^2$ norm error metric on various data.
Specifically, among the $21$ total cases, the algorithm finds optimal sample
sizes in $14$ cases. For the $10$ cases that the bootstrap is consistent, the
algorithm returns accurate results, and the error model fits the error profile
quite well. For other $11$ cases, such as Pareto2-Pareto3, even though the error
estimation method is inconsistent, and the $r^2$ score is only around $0.5$,
meaning that the model does not fit the profile perfectly, our algorithm also
manage to find the optimal sample size accurately. This again shows the strong
applicability of our algorithm in practice to almost all kinds of queries, even
to those that is theoretically inapplicable.

In conclusion, the two applicability experiments above demonstrate that the
\textsc{$L^2$Miss} algorithm can be applied to queries for which the error
estimation method, i.e., the bootstrap, is theoretically consistent. For other
queries that do not enjoy such a decent theoretical guarantee, few error
estimation methods can be applied in theory, while our algorithm may still work
well in practice. This gives users enough confidence to apply the algorithm to
build an AQP system that is general enough to be applied to almost all kinds of
analytical tasks, ranging from those as simple as finding averages to those as
complicated as regression and classification.

\subsection{Efficiency Evaluation}

In this section, we study the efficiency of our algorithms. We compare the
\textsc{$L^2$Miss} algorithm and its extensions against several state-of-the-art
sampling-based AQP algorithms based on different error estimation methods. The
AQP algorithms we evaluated include:

\begin{list}{\labelitemi}{\leftmargin=1em}\itemsep 0pt \parskip 0pt
    \item The \textsc{$L^2$Miss} algorithm described in \autoref{alg:l2Miss}
    based on our MISS framework and uses the bootstrap for error estimation.

    \item The \textsc{OrderMiss} algorithm described in \autoref{sec:ordering},
    which is also MISS-based and also employs the bootstrap to estimate errors.

    \item The Sample+Seek framework (SPS) proposed
    in~\cite{DBLP:conf/sigmod/DingHCC016}, which uses Chernoff-type bounds for
    error estimation.

    \item Our implementation of the sample selection algorithm of
    BlinkDB~\cite{DBLP:conf/eurosys/AgarwalMPMMS13} (BLK), which is based on the
    normality assumption to derive closed-form error estimation.

    \item The IFocus algorithm (IF) proposed
    in~\cite{DBLP:journals/pvldb/KimBPIMR15}, which uses Hoeffding's
    inequalities~\cite{Hoeffding1963} for error estimation.
\end{list}

Among the five algorithms above, \textsc{$L^2$Miss} is compared with SPS and BLK
since they all work with the $L^2$ norm error metric. \textsc{OrderMiss} is
compared with IF, since both provide ordering guarantees. For \textsc{$L^2$Miss}
and \textsc{OrderMiss}, since they require bootstrapping, we expect that their
running time will be larger given that the total sample sizes are almost the
same. SPS uses measure-biased sampling, which requires full scans on the data.
Therefore, the running time will grow as the data size increases. BLK uses
ad-hoc error estimation methods such that no estimation error exists. Therefore,
BLK can be viewed as the best method as long as it can be applied.

The dataset is the TPC-H~\cite{tpc-h} and the factors evaluated include the
relative error bound $\epsilon^*$, the error probability $\delta$, the number of
groups of the data $m$, and the size of the dataset $N$. We define each query to
have only one group-by attribute and one analytical attribute. In the queries,
the relative error bound is varies from $0.01$ to $0.002$ and the error
probability $\delta = 0.1, 0.05, 0.01, 0.005, 0.001$ to evaluate their impact on
performance. Moreover, tn order to evaluate the performance for multi-group
queries, the group-by attributes used are \texttt{LINESTATUS},
\texttt{RETURNFLAG}, \texttt{SHIPINSTRUCT}, \texttt{LINENUMBER}, and
\texttt{TAX} such that the numbers of groups of the data are $2$, $3$, $4$, $7$,
and $9$ respectively. The analytical attribute is \texttt{EXTENDEDPRICE}. Each
query involves only one group-by attribute and one analytical attribute. The
data size is changed by modifying the scale factor of TPC-H. Specifically, $N$
is approximately the scale factor times $6\times 10^{6}$. The scale factors used
include $1$, $10$, $30$, and $100$, which are officially designated by the TPC-H
specification.

To measure efficiency, as mentioned at the beginning of this section, we measure
both the running time and the total sample size $C\paren{\vec{n}}$ for each
algorithm. As we always concern about accuracy, we also compute and plot the
simulated confidence $\hat c$. To obtain more reliable results, when evaluating
one factor, we control the others to keep as default. The default values are:
$\epsilon^* = 0.01$, $\delta = 0.05$, $m = 1$, $N = 6 \times 10 ^ 6$， i.e., the
scale factor $= 1$, and $f = \texttt{AVG}$, which is supported by all the five
algorithms. For other parameters, $B = 500$, $I_n = \left[1000, 2000\right]$,
and $l = 5\paren{m + 1}$ where $m$ is the number of groups.

In the following experiments, all the algorithms are implemented by us in Python
with multi-core support. The data and the index are loaded into the memory
lazily (i.e., when they are accessed) for better performance. We use
\texttt{mmap(2)}~\cite{kerrisk2010linux} to randomly access specific rows
efficiently. All intermediate results are stored in memory to simplify our
experiments.

\subsubsection{{\large $L^2$} Norm Error}

First, we evaluate the performance of \textsc{$L^2$Miss} against SPS and BLK
under the $L^2$ norm error metric. The results are shown in
\autoref{fig:normal}. From the figure, we can learn that all algorithms achieve
satisfactory accuracy in all test cases as the simulated confidence is around or
above $1 - \delta$.

\begin{figure*}[ht]
    \centering
    \includegraphics{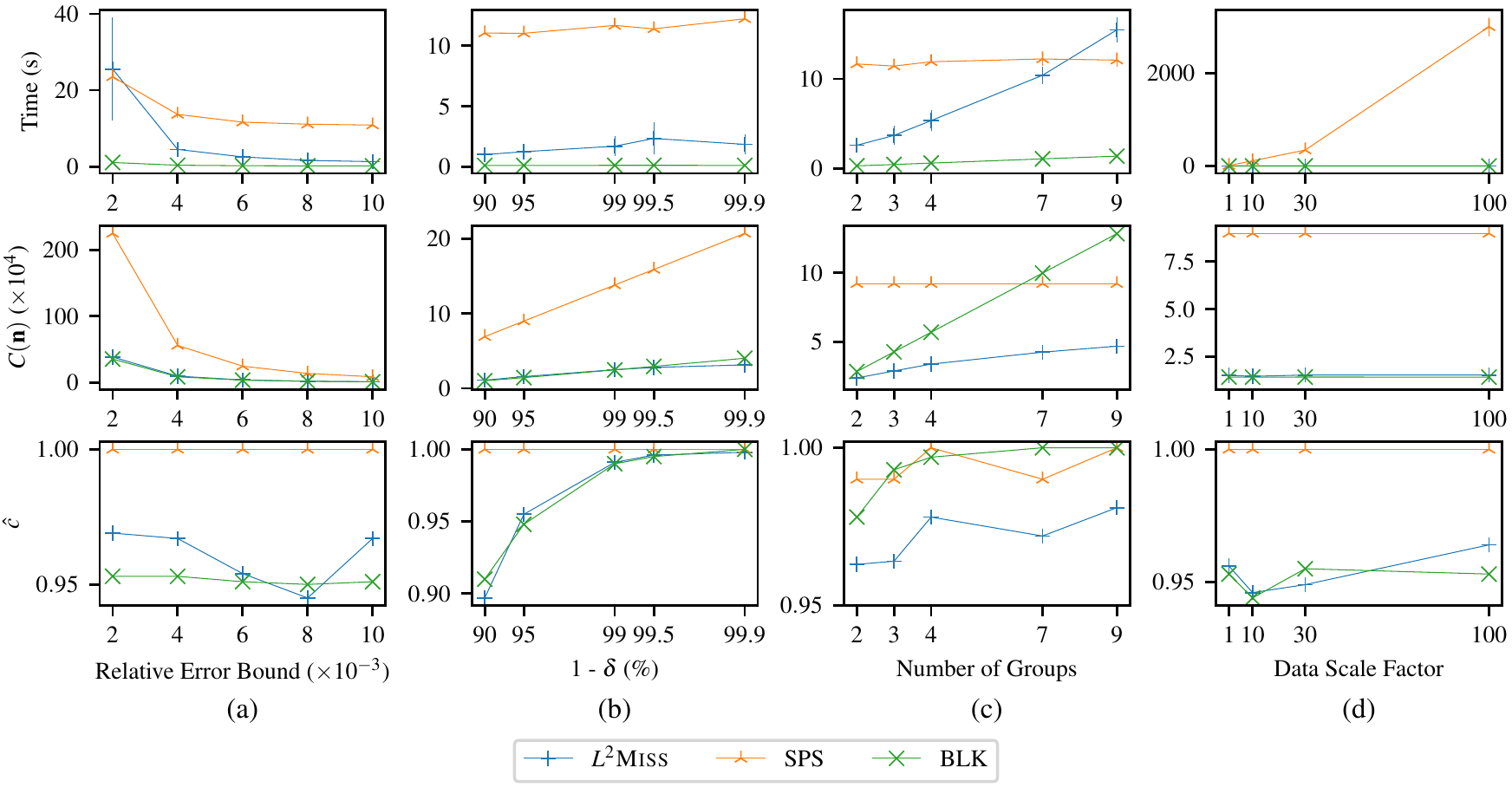}
    \caption{Efficiency evaluation for algorithms that work under the $L^2$ norm error metric}
\label{fig:normal}
\end{figure*}

\header{Relative error bounds.} \autoref{fig:normal}(a) shows the effect of
varying relative error bound $\epsilon^*$. The running time and the total sample
size of SPS is considerably larger than those of \textsc{$L^2$Miss} and BLK as
expected. However, even though the total sample size of \textsc{$L^2$Miss} and
BLK is quite similar, showing that \textsc{$L^2$Miss} is near-optimal in terms
of total sample size. However, the running time of \textsc{$L^2$Miss} grows
significantly faster than the others as the relative error bound decrease. This
also coincides with our expectation since the bootstrap in \textsc{$L^2$Miss} is
rather expensive. When the sampling rate is large, for example, at $~0.1$ when
$\epsilon^* = 0.002$, \textsc{$L^2$Miss} has no advantage over SPS.

\header{Error probabilities.} \autoref{fig:normal}(b) shows the performance with
different error probabilities. This is similar to the case of varying
$\epsilon^*$, i.e., \textsc{$L^2$Miss} is close to BLK, both in the running time
and the total sample size, while they are both more than 3x faster than SPS.
This shows that, when the sampling rate is small, for example, $\leq 0.01$,
bootstrap-based methods can still be very efficient compared with those
requiring full scans.

\header{Number of groups.} \autoref{fig:normal}(c) shows the performance on the
data of different numbers of groups. As the number of group increases, both the
running time and the total sample size of \textsc{$L^2$Miss} and BLK all
increase while they remain the almost the same for SPS. This is because that
\textsc{$L^2$Miss} and BLK consider groups separately, while SPS treats all the
groups as a whole. Furthermore, the total sample size of BLK is considerably
larger than that of \textsc{$L^2$Miss}. This is due to our naive implementation
of BLK that we let the errors of all groups be the same, while we let
\textsc{$L^2$Miss} determine the error distribution automatically to minimize
the total sample size.

\header{Data sizes.} \autoref{fig:normal}(d) shows the effect on efficiency for
different data sizes. As the data size grows, the running time and the total
sample size increase rapidly for SPS, while for \textsc{$L^2$Miss} and BLK, they
remain nearly unchanged. Note that when the data scale factor is $100$, the
assumption that all intermediate results can be fit into memory fails to hold
for SPS, resulting in a surge in the running time to around $3000$ seconds.
(Note that at this point we only draw several samples to compute the simulated
confidence for SPS since it is too time-consuming.) The reason is that SPS
requires access all the data, while \textsc{$L^2$Miss} and BLK only need to
access the sample data they draw. Therefore, when the sampling rate is small,
\textsc{$L^2$Miss} and BLK are much more efficient than SPS.

In conclusion, our \textsc{$L^2$Miss} algorithm indeed finds optimal sample
sizes compared with BLK and achieves satisfactory efficiency when the sampling
rate is small, e.g., $\leq 0.01$, when the user would like to trade more
accuracy for efficiency. On the other hand, SPS performs well when the sampling
rate is large, e.g., $> 0.1$, since it requires full scans. As for BLK, even
though its performance seems perfect for all cases, it only supports simple
aggregations such as \texttt{AVG}, \texttt{COUNT}, and \texttt{VAR}, while MISS
supports almost all kinds of analytical functions.

\subsubsection{Other Error Metrics}

In this section, we focus on evaluating whether our extensions effectively
supports other error metrics. Specifically, we focus on the performance of
\textsc{OrderMiss} providing ordering guarantees compared with IF. To the best
of our knowledge, IF is the only existing algorithm providing such guarantees.

Note that since different groups in the TPC-H dataset is rather identical in
terms of their analytical results. Therefore, guaranteeing correct ordering
properties would consume too much time and too large samples. Therefore, we add
a bias to each group such that the analytical results of any two groups differ
in a specific amount relative to the true analytical results, which are called
the group bias. The group bias is set to $0.05$ by default, meaning that the
difference of the analytical results in two adjacent groups in sorted order is
about $0.05$ times their true results.

\begin{figure}[ht]
    \centering
    \includegraphics[width=\linewidth]{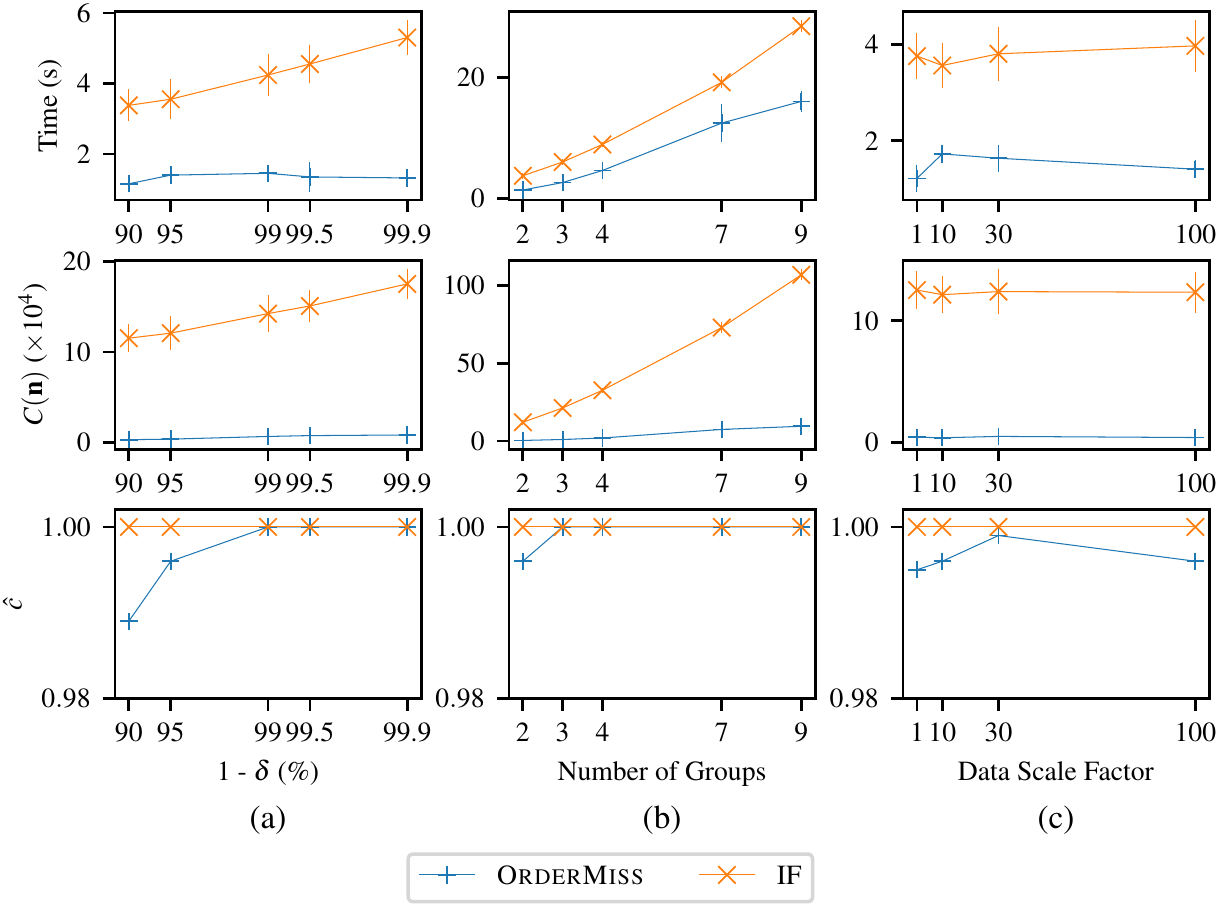}
    \caption{Efficiency evaluation for algorithms providing ordering guarantee}
\label{fig:order}
\end{figure}

The results are presented in \autoref{fig:order}. As observed from the third row
of \autoref{fig:order}, both \textsc{OrderMiss} and IF are able to produce
accurate results, i.e., the correct ordering property holds with the given
confidence. This indicates that our scheme, i.e., error bound conversion,
successfully extends \textsc{$L^2$Miss} to support other error metrics.

In terms of efficiency, we can learn from \autoref{fig:order} that for different
error probabilities, numbers of groups and data sizes, even though the two
algorithms behave similarly in trend, our \textsc{OrderMiss} algorithm
demonstrate its superiority against IF. Even though the bootstrap used by
\textsc{OrderMiss} is costly, \textsc{OrderMiss} is still faster than IF since
the total sample size of the former is several times smaller than the latter.

In summary, \textsc{OrderMiss} is much more efficient in providing ordering
guarantees compared with IF, yet still ensuring accuracy. This is because MISS
uses bootstrapping for error estimation, which is more accurate than
concentration inequalities used by IF, resulting in great reduction in total
sample size and the running time.

\section{Related Work}
\label{sec:related}
In this section, we survey the algorithms related to AQP briefly.

\header{Sample Selection for AQP.} Since
BlinkDB~\cite{DBLP:conf/eurosys/AgarwalMPMMS13} was ever proposed,
sampling-based AQP systems have been a hotspot in recent years in both industry
and academia. Notable systems besides BlinkDB include
iOLAP~\cite{DBLP:conf/sigmod/ZengAS16},
Sample+Seek~\cite{DBLP:conf/sigmod/DingHCC016}, and
SnappyData~\cite{DBLP:conf/cidr/MozafariRMMCBB17}. However, only those who
employ closed-form error estimation methods select suitable samples
automatically, which is far from optimal or even inapplicable for many queries.
Others, which use numerical error estimation methods, e.g., the bootstrap, do
not equip with such mechanism, and users are required to create or select their
desired samples manually. For example,
SnappyData\cite{DBLP:conf/cidr/MozafariRMMCBB17}, which adopts the latter
approach, requires users to create samples manually. It would select the largest
one if more than one samples are available~\cite{snappydata}. This approach
limits its use to only experts. Our approach, on the contrary, determines the
optimal sample size and draws the sample for each given query automatically,
which makes it much easier to use.

\balance
\header{AQP for complex queries.} As noted in \autoref{sec:problem},  we only
consider simple analytical queries without selection and join in general in this
paper. However, these two types of queries are essential to build a practical
AQP system. Fortunately, our MISS framework is flexible enough to incorporate
state-of-the-art techniques to gain such abilities. For joins,
WanderJoin~\cite{DBLP:conf/sigmod/0001WYZ16} is also a sampling-based technique
that estimates the sampling probability for each tuples and relies on
closed-form methods to estimate the approximation error for a analytical query.
This method can be perfectly collaborated with our approach to produce overall
approximate results directly to the end user. For selection,
Sample+Seek~\cite{DBLP:conf/sigmod/DingHCC016} proposes an indexing scheme that
accelerate selections in general using inverted indexing. This method is also
compatible with our MISS framework and therefore can be directly incorporated
into MISS to enable the MISS-based algorithms sample through an pre-built index.

\header{AQP without sampling.} Other AQP approaches without sampling include
data cubes~\cite{DBLP:conf/sigmod/HellersteinHW97} and
sketches\cite{DBLP:journals/cacm/Cormode17} are also widely used. Nevertheless,
these approaches are usually particularly designed for specific scenarios and
have their own limitations. For example, data cubes are used when both the data
and the query would not change much over time such that the pre-computed results
are still valid. Sketches are for data streams and can only be applied to
specific types of queries. On the contrary, our MISS framework and its
derivative algorithms are designed to be suitable for as many types of queries
as possible. By using online sampling and bootstrapping, our approach make
almost no assumptions on the data and the query performed.

\section{Conclusion and Future Work}
\label{sec:conclusion}

In this paper, we propose the Model-based Iterative Sample Selection (MISS)
frameworks and a family of algorithms based on the framework to find optimal
sample sizes for various error metrics based on the error model that we build.
Our approach not only supports $L^2$ norm error metric, but also provides other
types of guarantees including bounded maximum error, bounded difference error
and correct ordering under the MISS framework. With minor modifications, our
algorithm can support various types of total sample size including uniformly or
non-uniformly linear, polynomial and even exponential. We show theoretically and
empirically that our algorithms can find near-optimal samples that satisfy
user-defined error constraints for almost all kinds of data and user-defined
analytical functions.

For future work, we see two directions that worth exploring. The first is to
establish rigorous theories on what specific kinds of queries that our error
model can be applied to, or even to find better error models. And the second is
to build a practical and easy-to-use AQP system that incorporates our approaches
and other techniques discussed in \autoref{sec:related} to support arbitrary
queries facing directly to end users.

%\section{Acknowledgments} The authors would like to thank Professor Michael I.
%Jordan from UC Berkeley for the invaluable discussion with him on asymptotic
%statistics. Xuebin would like to thank Kunyi Chen for her insightful
%suggestions that greatly improved this work. This research is supported by xxx.

\balance

\bibliographystyle{abbrv}
\bibliography{mass}

\end{document}